\pdfoutput=1
\documentclass[draftcls,onecolumn, 11pt]{IEEEtran} 
\usepackage{amsmath,amssymb,graphicx,amsfonts,subfigure,amsthm}
\usepackage{epstopdf}
\usepackage[T1]{fontenc}
\usepackage[latin9]{inputenc}
\usepackage{times,verbatim,algorithm,algorithmic,cite}

\usepackage{etoolbox}  
\patchcmd{\thebibliography}{\leftmargin\labelwidth}
     {\itemsep 0pt \leftmargin\labelwidth}{}{}

\DeclareMathOperator*{\argmin}{arg\,min}

\newtheorem{theorem}{Theorem}
\newtheorem{lemma}{Lemma}

\begin{document}
\title{Optimum Design for Coexistence Between Matrix Completion Based MIMO Radars and a MIMO Communication System}
%
\author{Bo~Li, ~Athina~P.~Petropulu,  and Wade Trappe,  \thanks{This work is supported by NSF under Grant ECCS-1408437. Parts of this work have been presented at the IEEE International Conference on Acoustics, Speech, and Signal
Processing (ICASSP'15) \cite{LiBo15}.

The authors are with Department of Electrical and Computer Engineering, Rutgers, The State University of New Jersey, Piscataway NJ 08854, USA. (E-mails:  \{paul.bo.li,athinap\}@rutgers.edu; trappe@winlab.rutgers.edu)}}

\maketitle
\begin{abstract}
Recently proposed multiple input multiple output radars based on matrix completion (MIMO-MC)  employ sparse sampling  to reduce the amount of data that need to be forwarded to the radar fusion center, and as such enable savings in communication power and bandwidth.
This paper proposes designs that optimize the sharing of spectrum between a MIMO-MC radar and a
communication system, so that the latter interferes minimally with the former.
First, the communication system transmit covariance matrix is designed to minimize the effective interference power (EIP) to the radar receiver, while maintaining certain average capacity and transmit power for the communication system.
Two approaches are proposed, namely a noncooperative and a cooperative approach, with the latter being applicable when the radar sampling scheme is known at the communication system.
Second, a joint design of the communication transmit covariance matrix and the MIMO-MC radar sampling scheme is proposed, which achieves even further EIP reduction.
\end{abstract}
\begin{keywords}
Collocated MIMO radar, matrix completion, spectrum sharing
\end{keywords}

\section{Introduction}
The operating frequency bands of communication  and radar systems often overlap,  causing one system to exert interference to the other. For example, the high UHF radar systems overlap with GSM communication systems, and the S-band radar systems partially overlap with Long Term Evolution (LTE), and WiMax systems \cite{Web13,Sanders12,Lackpour11,Sodagari12}.
Spectrum sharing is a  new line of work that targets at enabling radar and communication systems to share the spectrum efficiently by minimizing interference effects \cite{Lackpour11,Sodagari12,Babaei13,Deng13,Amuru13,Khawar14}.

This paper investigates the  problem of spectrum sharing between a MIMO communication system and a matrix completion (MC) based colocated MIMO radar (MIMO-MC) system \cite{Sun13,Denis14,Sun14}. MIMO radars transmit orthogonal waveforms from their multiple transmit (TX) antennas, and their receive (RX) antennas forward their measurements  to a fusion center for further processing.
The RX antenna measurements could be samples of the target returns, or could be the outputs of matched filters.
Based on the forwarded data, the fusion center populates a matrix, referred to as the ``data matrix", which is then used  by standard array processing schemes for  target estimation. When the target returns are sampled at the Nyquist rate, and for a relatively small number of targets, the data matrix is low-rank \cite{Sun13}, thus, under certain conditions it can be reconstructed based on a small, uniformly sampled set of its entries.
This observation is the basis of MIMO-MC radars \cite{Sun13,Denis14,Sun14}, in which the RX antennas forward to the fusion center a small number of pseudo-randomly obtained samples of the target returns, or the result of matched filtering with a set of randomly selected transmit waveforms, along with information on the sampling scheme, with each RX antenna partially filling a row of the  data matrix. Subsequently, the full data matrix is recovered using MC techniques. MIMO-MC radars maintain the high resolution of MIMO radars, while requiring significantly fewer data to be communicated to the fusion center, thus enabling savings in communication power and bandwidth. These savings are especially important when the RX antennas are on battery operated nodes, and/or the communication to the fusion center occurs in a wireless fashion. Compared to compressive sensing (CS) based MIMO radars, MIMO-MC radars achieve data reduction while avoiding the basis mismatch issues which were inherent in CS-based approaches \cite{Yu10}.

In this paper, the MIMO-MC radar system is considered as the primary user of the channel, while the MIMO communication system is the secondary user.
{\it First}, for a fixed uniformly random radar sub-sampling scheme, the communication system optimally designs its transmit covariance matrix so that its effective interference power (EIP) exerted to the radar RX node is minimized, while its own average capacity and transmit power are kept at a prescribed level.
In doing so, two approaches are proposed, namely, a cooperative and a noncooperative approach, with the latter being applicable when the communication system has knowledge of the MIMO-MC radar sampling instances. It is shown that when the MIMO-MC radar sampling scheme is known to the communication system, the EIP can be greatly reduced, especially at low sub-sampling rates.
{\it Second}, a joint-design of the radar sampling scheme and the communication system transmit covariance matrix  is proposed, targeting at minimizing the EIP at the radar RX node.
Alternating optimization is employed to solve the optimization problem.
The candidate sampling scheme needs to be such that the resulting data matrix can be completed. Recent work
\cite{Srinadh14} showed that for matrix completion, the sampling locations should correspond to a binary matrix with large spectral gap. Since the spectral gap of a matrix is not affected by column and row permutations, we propose to search for the optimum sampling matrix among matrices which are row and column permutations of an initial  sampling matrix with large spectrum gap.
Even before any design is implemented, the MIMO-MC radar system is expected to be less susceptible to interference than a plain MIMO radar; this is because the interference affects only some entries of the data matrix. As it is shown in the paper, by appropriately designing the communication TX waveforms and/or the radar sampling scheme, the interference can be further reduced.

The paper is organized as follows. Section \ref{sec:sigmodel} introduces the signal model when the MIMO-MC radar and communication systems coexist. The problem of a MIMO communication system sharing the spectrum with a MIMO-MC radar system is studied in Section \ref{sec:schemeI} and \ref{sec:schemeII}. Numerical results, discussions and conclusions are provided in Section \ref{sec:discussions}-\ref{sec:conclusion}.\\
{\it Notation:} $\mathcal{CN}(\mathbf{\mu},\mathbf{\Sigma})$ denotes the circularly symmetric complex Gaussian distribution with mean $\mathbf{\mu}$ and covariance matrix $\mathbf{\Sigma}$. $|\cdot|$ and $\text{Tr}(\cdot) $ denotes the matrix determinant and trace, respectively. The set $\mathbb{N}_{L}^+$ is defined as $\{1,\dots, L\}$. $\mathcal{N}(\mathbf{A})$ and $\mathcal{R}(\mathbf{A})$ denote the null and row spaces of matrix $\mathbf{A}$, respectively. $\mathbf{A}_{i\cdot }$ and $\mathbf{A}_{\cdot j}$ respectively denote the $i$-th row and $j$-th column of matrix $\mathbf{A}$. $[\mathbf{A}]_{i,j}$ denotes the element on the $i$-th row and $j$-th column of matrix $\mathbf{A}$. $x^+$ is defined as $\max(0,x)$.

\section{Background on MIMO-MC Radars}
\label{sec:background}
Consider a colocated MIMO radar system with $M_{t,R}$ TX antennas and $M_{r,R}$ RX antennas, targeting at the estimation of far-field targets. The radar operates in two phases; in the first phase the TX antennas transmit waveforms and the RX antennas receive target returns, while in the second phase, the RX antennas forward their measurements to a fusion center.
The $m$-th, $m\in\mathbb{N}_{M_{t,R}}^+$ antenna transmits a coded waveform containing $L$ symbols $\{s_m(1),\cdots,s_m(L)\}$ of duration $T_R$ each. Suppose that each RX antenna samples the target returns with sampling interval $T_R$, {\it i.e.}, each symbol in the waveform is sampled exactly once. The sampling time instances are given as $\{T_R,\cdots,LT_R\}$.
Following the model in \cite{Sun13,Sun14,Denis14}, the data matrix received at the RX antennas is formulated as
\begin{equation}\label{eqn:sigmodelMIMO}
\mathbf{Y}_R={\gamma\rho}\mathbf{DS}+\mathbf{W}_R,
\end{equation}
where $\gamma$ and $\rho$ respectively denote the path loss corresponding to the range bin of interest, and the radar transmit power; $\mathbf{D}\in\mathbb{C}^{M_{r,R}\times M_{t,R}}$ denotes the target response matrix, which depends on the target reflectivity, angle of arrival and target speed (details can be found in \cite{Sun14}); $\mathbf{S}=[\mathbf{s}(1),\cdots,\mathbf{s}(L)]$, with $\mathbf{s}(l)=[s_1(l),\cdots,s_{M_{t,R}}(l)]^T$ being the sampled waveform matrix. The transmit waveforms are typically orthogonal, thus it holds $\mathbf{SS}^H=\mathbf{I}$ \cite{Sun14}. $\mathbf{Y}_R\triangleq [\mathbf{y}_R(1),\dots,\mathbf{y}_R(L)]$; $\mathbf{W}_R\triangleq [\mathbf{w}_{R}(1),\dots,\mathbf{w}_{R}(L)]$ is the additive noise matrix.

Matrix $\mathbf{D}$ has rank equal to the number of targets thus, it is low-rank if the number of targets is much smaller than $M_{r,R}$ and $M_{t,R}$. Similarly, matrix $\mathbf{DS}$ is low-rank if the number of targets is much smaller than $M_{r,R}$ and $L$.
The RX antennas of the matrix completion based MIMO (MIMO-MC) radar \cite{Sun13,Sun14,Denis14} subsample the target returns and forward the samples, along with the corresponding sampling times to the fusion center, thus partially populating the data matrix. The full data matrix is then completed with matrix completion techniques, and target estimation can be implemented based on the completed matrix via standard array processing schemes \cite{Krim96}.

The partially filled data matrix can be mathematically expressed as follows \cite{Sun13,Sun14}
\begin{equation*}\label{eqn:sigmodelMCI}
\mathbf{\Omega}_I\circ \mathbf{Y}_R=\mathbf{\Omega}_I\circ ({\gamma\rho}\mathbf{DS}+\mathbf{W}_R), \quad\quad \text{(Scheme I)}
\end{equation*}
where $\circ$ denotes Hadamard product and $\mathbf{\Omega}_I$ is a matrix with ``$0$"s or ``$1$"s, with the "$1$"s corresponding to the sampling instances. In the physical implementation, only the entries of $\mathbf{Y}_R$ corresponding to ``$1$"s in $\mathbf{\Omega}_I$ represent obtained samples.
The sub-sampling rate, $p_{I}$, equals $\lVert\mathbf{\Omega}_I\rVert_0/LM_{r,R}$. The above MIMO-MC scheme is referred to in \cite{Sun13,Sun14} as Scheme I.

Alternatively, a random matched filter bank (RMFB) at each RX antenna generates a data matrix which can be expressed as \cite{Denis14}
\begin{equation*}\label{eqn:sigmodelMCII}
\mathbf{\Omega}_{II}\circ (\mathbf{Y}_R\mathbf{S}^H)=\mathbf{\Omega}_{II}\circ ({\gamma\rho}\mathbf{D}+\mathbf{W}_R\mathbf{S}^H), \quad\quad \text{(Scheme II)}
\end{equation*}
where $\mathbf{\Omega}_{II}$ is a sampling matrix with binary entries and dimension $M_{r,R}\times M_{t,R}$. The locations of ``$1$"s at the $m$-th row are the indices of the matched filters that were used at the $m$-th RX antenna, {\it i.e.}, $\xi_m\subset \mathbb{N}_{M_{t,R}}^+$.
The sub-sampling rate $p_{II}$ is defined as $\lVert\mathbf{\Omega}_{II}\rVert_0/M_{t,R}M_{r,R}$. This MIMO-MC radar scheme is referred to in \cite{Denis14} as Scheme II.

Early studies on matrix completion theory suggested that the low-rank matrix reconstruction from partial entries succeeds with high probability if the low-rank matrix satisfies the incoherence property \cite{Candes10}, and the entries are sampled uniformly at random.
However, recent works \cite{Srinadh14} showed that, regarding the sampling of elements, it is sufficient that the sampling matrix has large spectral gap ({\it i.e.}, large gap between the largest and second largest singular values).
In \cite{Sun13,Denis14,Sun14} that the matrix $\mathbf{DS}$ exhibits low coherence while the sampling of its elements was a result of uniformly random sampling at the RX antennas.

\section{System Model}
\label{sec:sigmodel}
Consider a MIMO communication system  which coexists with a MIMO-MC radar system
as shown in Fig. \ref{fig:diagram}, sharing the same carrier frequency.
The MIMO-MC radar operates in two phases, {\it i.e.}, in Phase 1 the RX antennas obtain measurements of the  target returns, and in Phase 2, the RX antennas forward the obtained samples to a fusion center. The communication system interferes with the radar system during both phases. In the following, we will address spectrum sharing during the first phase only.
The interference during the second phase can be viewed as the interference between two communication systems, and addressing this problem has been covered in the literature \cite{Zhang08,Zhang10}.

In the following, Scheme I is used to illustrate the system model. Suppose that the two systems have the same symbol rate and are synchronized in sampling time (see Section \ref{sec:discussions} for the mismatched case). We do not assume perfect carrier phase synchronization between the two systems. The data matrix corresponding to the radar system and the received matrix at the communication RX antennas during $L$ symbol durations can be respectively expressed as
\begin{equation}\label{eqn:sigmodelmatR} 
\begin{aligned}
\mathbf{\Omega}_{I}\circ \mathbf{Y}_R={\mathbf{\Omega}_{I}}\circ ({\gamma\rho}\mathbf{D}\mathbf{S}+\mathbf{G}_2\mathbf{X}\mathbf{\Lambda}_2+\mathbf{W}_R),
\end{aligned}
\end{equation}
\begin{equation}
\begin{aligned}
\mathbf{Y}_C=\mathbf{H}\mathbf{X}+{\rho}\mathbf{G}_1\mathbf{S}\mathbf{\Lambda}_1+\mathbf{W}_C,
\end{aligned}
\end{equation}
where
\begin{itemize} \itemsep -1pt
\item $\mathbf{Y}_R$, $\rho,\mathbf{D}$, $\mathbf{S}$, $\mathbf{W}_R$, and $\mathbf{\Omega}_I$ are defined in Section \ref{sec:background}.
\item $\mathbf{X}\triangleq [\mathbf{x}(1),\dots,\mathbf{x}(L)]$; $\mathbf{Y}_C\triangleq [\mathbf{y}_C(1),\dots,\mathbf{y}_C(L)]$; $\mathbf{W}_{C}\triangleq [\mathbf{w}_{C}(1),\dots,\mathbf{w}_{C}(L)]$.
\item $\mathbf{y}_{C}(l)$ and $\mathbf{w}_{C}(l)$ respectively denote the signal and the additive noise at the radar/communication RX antennas sampled at the $l$-th sampling time. It is assumed that $\mathbf{w}_{C}(l)\sim \mathcal{CN}(0,\sigma_{C}^2\mathbf{I})$ and $\mathbf{w}_{R}(l)\sim \mathcal{CN}(0,\sigma_{R}^2\mathbf{I})$.
\item $\mathbf{H}\in\mathbb{C}^{M_{r,C}\times M_{t,C}}$ denotes the communication channel, where $M_{r,C}$ and $M_{t,C}$ denote respectively  the number of RX and TX antennas of the communication system \cite{Zhang08}; $\mathbf{G}_1\in\mathbb{C}^{M_{r,C}\times M_{t,R}}$ denotes the interference channel from the radar TX antennas to the communication system RX antennas \cite{Sodagari12,Babaei13,Khawar14}; $\mathbf{G}_2\in\mathbb{C}^{M_{r,R}\times M_{t,C}}$ denotes the interference channel from the communication TX antennas to the radar RX antennas. It is assumed that the channels remain the same over $L$ symbol durations.
\item $\mathbf{s}(l)$ and $\mathbf{x}(l)$ respectively denote the transmit vector at the radar and the communication TX antennas during the $l$-th symbol duration. The rows of $\mathbf{X}$ are codewords from the code-book of the communication system.
\item $\mathbf{\Lambda}_1$ and $\mathbf{\Lambda}_2$ are diagonal matrices. The $l$-th diagonal entry of $\mathbf{\Lambda}_1$, {\it i.e.}, $e^{j\alpha_{1l}}$, denotes the random phase offset between the MIMO-MC radar carrier and the communication receiver reference carrier at the $l$-th sampling time. The $l$-th diagonal entry of $\mathbf{\Lambda}_2$, {\it i.e.}, $e^{j\alpha_{2l}}$, denotes the random phase offset between the communication transmitter carrier and the MIMO-MC radar reference carrier at the $l$-th sampling time. The phase offsets result from the random phase jitters of the radar oscillator and the oscillator at the communication receiver Phase-Locked Loops. In the literature \cite{Gardner05,Poore01,Mudumbai07}, the phase jitter $\alpha(t)$ is modeled as a zero-mean Gaussian process. In this paper, we model $\{\alpha_{1l}\}_{l=1}^L$ as a sequence of zero-mean Gaussian random variables with variance $\sigma_{\alpha}^2$. Modern CMOS oscillators exhibit very low phase noise, e.g., $-94$ dB below the carrier power per Hz ({\it i.e.}, $-94$dBc/Hz) at an offset of $2\pi\times 1$ MHz, which yields phase jitter variance $\sigma_{\alpha}^2\approx 2.5\times10^{-3}$ \cite{Razavi96}.

\end{itemize}

It is assumed that the MIMO channels $\mathbf{H}$, $\mathbf{G}_1$ and $\mathbf{G}_2$ are perfectly known at the communication TX antennas.
In practice, the channel state information can be obtained through the transmission of pilot signals \cite{Sodagari12,cpc09}.
Based on knowledge of radar waveforms and $\mathbf{G}_1$, the communication system can reject some interference due to the radar via subtraction. However, due to the high power of the radar \cite{Sanders12} and the unknown phase offset, there will still be interference in the communication received signal, {\it i.e.},
\begin{equation*}
{\rho}\mathbf{G}_1\mathbf{S}(\mathbf{\Lambda}_1-\mathbf{I})\approx {\rho}\mathbf{G}_1\mathbf{S}\mathbf{\Lambda}_{\alpha},
\end{equation*}
where $\mathbf{\Lambda}_{\alpha}=\text{diag}(j\alpha_{11},\cdots,j\alpha_{1L})$, and the approximation is based on the fact that $\{\alpha_{1l}\}_{l=1}^L$ are small. The signal at the communication receiver after interference cancellation equals
\begin{equation} \label{eqn:sigmodelmatC}
\tilde{\mathbf{Y}}_C = \mathbf{H}\mathbf{X}+{\rho}\mathbf{G}_1\mathbf{S}\mathbf{\Lambda}_{\alpha}+\mathbf{W}_C.
\end{equation}
We observe that the residual interference is not circularly symmetric. The communication channel capacity is achieved by non-circularly symmetric Gaussian codewords, whose covariance and complementary covariance matrix are required to be designed simultaneously \cite{Taubock12}.
Here we consider the circularly symmetric complex Gaussian codewords $\mathbf{x}(l)\sim \mathcal{CN}(0,\mathbf{R}_{xl})$, which achieve a lower bound of the channel capacity. The design complexity is reduced since we only need to design the transmit covariance matrix $\mathbf{R}_{xl}$.

The communication system aims at minimizing its interference to the MIMO-MC radar, while maintaining its average capacity over $L$ symbol durations, by adapting its transmit resources in both time and spatial domain. In the following two sections, the spectrum sharing problem is formulated for both Schemes I and II.

\section{Spectrum Sharing with Scheme I Radars} \label{sec:schemeI}
In this section, we design the communication transmit waveforms, and in particular their covariance matrix, so that we minimize the interference power at the Scheme I radar RX node, while satisfying the communication rate and power constraints of the communication system.
The total transmit power of the communication TX antennas equals
$$
\begin{aligned}
\mathbb{E}\{\text{Tr}(\mathbf{XX}^H)\}=\mathbb{E}\left\{\text{Tr}\left(\sum_{l=1}^L \mathbf{x}(l)\mathbf{x}^H(l)\right)\right\} =\sum_{l=1}^L\text{Tr}(\mathbf{R}_{xl}),
\end{aligned}
$$
where $\mathbf{R}_{xl}\triangleq \mathbb{E}\{\mathbf{x}(l)\mathbf{x}^H(l)\}$.

\begin{figure}
  \centering
  \includegraphics[width=8cm]{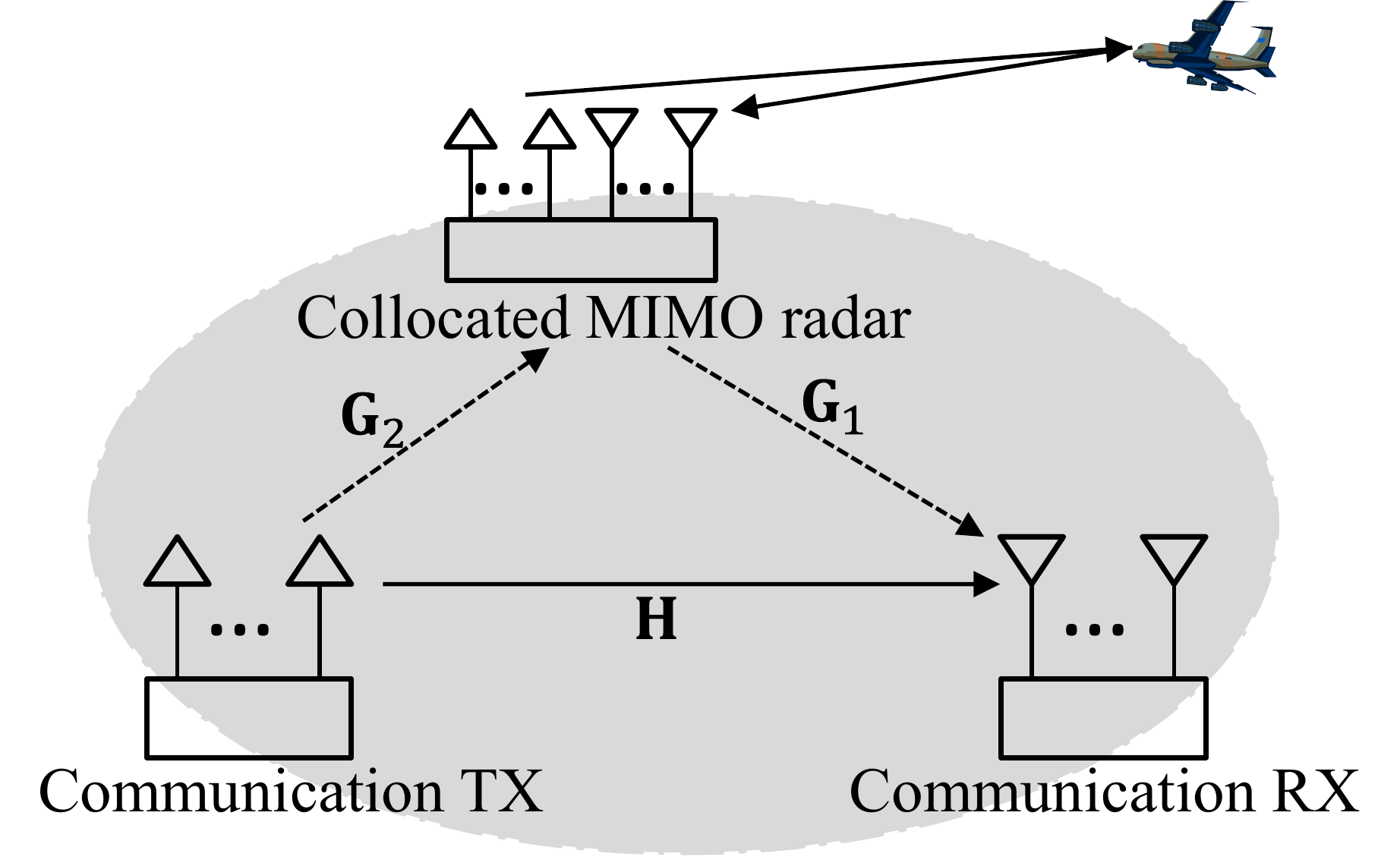}\\
  \vspace{-2mm}
  \caption{A MIMO communication system sharing spectrum with a colocated MIMO radar system}\label{fig:diagram}
\vspace{-2mm}
\end{figure}

According to (\ref{eqn:sigmodelmatR}), the total interference power (TIP) exerted at the radar RX antennas equals
\begin{equation}\label{eqn:interfarrive}
\begin{aligned}
\text{TIP}& \triangleq \mathbb{E}\{\text{Tr}(\mathbf{G}_2\mathbf{X}\mathbf{\Lambda}_2\mathbf{\Lambda}_2^H\mathbf{X}^H\mathbf{G}^H_2)\} \\
&=\sum\nolimits_{l=1}^L \text{Tr}\left( \mathbf{G}_2\mathbf{R}_{xl}\mathbf{G}_2^H \right).
\end{aligned}
\end{equation}
Since the radar only forwards part of $\mathbf{Y}_R$  to the fusion center, only the term ${\bf \Omega}_{I}\circ(\mathbf{G}_2\mathbf{X}\mathbf{\Lambda}_2)$
 represents effective interference to the radar system.
Based on this observation, we define the {\em effective interference power} (EIP) at the radar RX node as
\begin{small}
\begin{equation}\label{eqn:interfI}
\begin{aligned}
&\text{EIP}_{I} \triangleq \mathbb{E}\left\{\text{Tr}\left(\mathbf{\Omega}_{I}\circ(\mathbf{G}_2\mathbf{X}\mathbf{\Lambda}_2) \left(\mathbf{\Omega}_{I}\circ( \mathbf{G}_2\mathbf{X}\mathbf{\Lambda}_2)\right)^H \right)\right\}\\
=&\mathbb{E}\left\{\text{Tr}\left([\mathbf{G}_{21}\mathbf{x}(1)\dots\mathbf{G}_{2L}\mathbf{x}(L)]\mathbf{\Lambda}_2\mathbf{\Lambda}_2^H [\mathbf{G}_{21}\mathbf{x}(1)\dots \mathbf{G}_{2L}\mathbf{x}(L)]^H \right)\right\}\\
=&\mathbb{E}\left\{\text{Tr}\left(\sum_{l=1}^L \mathbf{G}_{2l}\mathbf{x}(l)\mathbf{x}^H(l)\mathbf{G}^H_{2l}]  \right)\right\} \\
=&\sum_{l=1}^L \text{Tr}\left( \mathbf{G}_{2l}\mathbf{R}_{xl}\mathbf{G}^H_{2l} \right)
=\sum_{l=1}^L \text{Tr}\left( \mathbf{\Delta}_l\mathbf{G}_{2}\mathbf{R}_{xl}\mathbf{G}^H_{2} \right),
\end{aligned}
\end{equation}
\end{small}
where $\mathbf{G}_{2l}\triangleq \mathbf{\Delta}_l\mathbf{G}_2$, with  $\mathbf{\Delta}_l$ being a diagonal matrix whose diagonal is $\mathbf{\Omega}_{\cdot l}$, i.e., $\mathbf{\Delta}_l=\text{diag}(\mathbf{\Omega}_{\cdot l})$.
We note that the EIP at sampling time $l$ contains the interference corresponding to ``$1$"s in $\mathbf{\Omega}_{\cdot l}$ only. It is equivalent to say that the effective interference channel during the $l$-th symbol duration is $\mathbf{G}_{2l}$.

In the coexistence model of (\ref{eqn:sigmodelmatR}) and (\ref{eqn:sigmodelmatC}), both the effective interference channel $\mathbf{G}_{2l}$ and interference power at the communication receiver $\mathbf{R}_{\text{int}l}\triangleq \rho^2\sigma_{\alpha}^2\mathbf{G}_1\mathbf{s}(l)\mathbf{s}^H(l)\mathbf{G}^H_1$ vary between sampling times. The communication system needs to use different covariance matrices for each symbol, {\it i.e.}, $\mathbf{R}_{xl}$, in order to match the variation of $\mathbf{G}_{2l}$ and $\mathbf{R}_{\text{int}l}$ and minimize the effective interference to the radar system while maintaining the capacity. The channel can be equivalently viewed as a fast fading channel with perfect channel state information at both the transmitter and receiver \cite{Goldsmith03,Tse05}. Similar to the definition of ergodic capacity \cite{Goldsmith03}, the achieved capacity is the average over $L$ symbols, {\it i.e.},
\begin{equation}\label{eqn:avgcapacity}
\text{AC}(\{\mathbf{R}_{xl}\})\triangleq \frac{1}{L}\sum\nolimits_{l=1}^L \log_2 \left|\mathbf{I} + \mathbf{R}_{wl}^{-1}\mathbf{HR}_{xl}\mathbf{H}^H \right|,
\end{equation}
where $\{\mathbf{R}_{xl}\}$ denotes the set of all $\mathbf{R}_{xl}$'s and $\mathbf{R}_{wl}\triangleq \mathbf{R}_{\text{int}l} + \sigma_C^2\mathbf{I}$ for all $l\in\mathbb{N}_L^+$.

In the following we will consider three spectrum sharing approaches between the communication and Scheme I radar, namely, a noncooperative, a cooperative and a joint design approach.
In the cooperative and joint design approaches, the communication system knows the radar sampling scheme.
The performance improvement is expected to be higher under higher level of cooperation at the cost of reduced
security and increased coordination complexity.
\subsection{Noncooperative Spectrum Sharing} \label{sec:noncooperI}
In the {\em noncooperative approach}, the communication system has no knowledge of $\mathbf{\Omega}_{I}$. Therefore, it cannot obtain the expression of $\text{EIP}_{I}$ of (\ref{eqn:interfI}). In this case, the communication system will design its covariance matrix to minimize the TIP in (\ref{eqn:interfarrive}) as follows:
\begin{subequations}
\begin{align}
({\bf P}_0)\quad \min_{\{\mathbf{R}_{xl}\}\succeq 0} \; \text{TIP} (\{\mathbf{R}_{xl}\}) 
\;\text{s.t. }\, & \sum\nolimits_{l=1}^L \text{Tr}\left(\mathbf{R}_{xl}\right) \le P_t \label{eqn:constrPt} \\
& \text{AC}(\{\mathbf{R}_{xl}\}) \ge C, \label{eqn:constrSR}
\end{align}
\end{subequations}
where the constraint of (\ref{eqn:constrPt}) restricts the total  transmit power at the communication TX antennas to be no larger than $P_t$. The constraint of (\ref{eqn:constrSR}) restricts the communication average capacity during $L$ symbol durations to be at least $C$, in order to provide reliable communication and avoid service outage. $\{\mathbf{R}_{xl}\}\succeq 0$ imposes the positive semi-definiteness on the solution. Let us denote by $\mathbb{X}_0$ the feasible set determined by the above three constraints.
Problem (${\bf P}_0$) is convex.

The power constraints of (\ref{eqn:constrPt}) and (\ref{eqn:constrSR}) are jointly applied for all $L$ symbol durations. The extension to constraints individually applied for each symbol duration is straightforward because the convexity of the problem is preserved.
Problem (${\bf P}_0$) is a variant of the Problem ($\mathbf{P}_6$) in \cite{Zhang08} for multichannel spectrum sharing in cognitive radio network.
\subsection{Cooperative Spectrum Sharing} \label{sec:cooperI}
In the {\em cooperative approach}, the MIMO-MC radar shares its sampling scheme $\mathbf{\Omega}_{I}$ with the communication system. Now, the spectrum sharing problem can be formulated as
\begin{equation}
\begin{aligned}
({\bf P}_1)\quad \min \; \text{EIP}_{I} (\{\mathbf{R}_{xl}\}) 
\;\text{s.t. }\, & \{\mathbf{R}_{xl}\} \in \mathbb{X}_0.
\end{aligned}
\end{equation}
Problem ($\mathbf{P}_1$) has exactly the same constraints as (${\bf P}_0$).

The Lagrangian of ($\mathbf{P}_1$) can be written as
\begin{equation*}
\begin{aligned}
\mathcal{L}(\{\mathbf{R}_{xl}\},\lambda_1,\lambda_2)= &\text{EIP}_{I} (\{\mathbf{R}_{xl}\})  + \lambda_2\left(C-\text{AC}(\{\mathbf{R}_{xl}\})\right)\\
+& \lambda_1\left(\sum\nolimits_{l=1}^L \text{Tr}\left(\mathbf{R}_{xl}\right) - P_t\right),
\end{aligned}
\end{equation*}
where $\lambda_1\ge 0$ is the dual variable associated with the transmit power constraint, and $\lambda_2 \ge 0$ is the average capacity constraint.
The dual problem of ($\mathbf{P}_1$) is given as
\begin{equation*}
({\bf P}_1\text{-}{\bf D})\; \max_{\lambda_1,\lambda_2 \ge 0} g(\lambda_1,\lambda_2),
\end{equation*}
where $g(\lambda_1,\lambda_2)$ is the dual function defined as
\begin{equation*}
g(\lambda_1,\lambda_2)=\inf_{\{\mathbf{R}_{xl}\}\succeq 0} \mathcal{L}(\{\mathbf{R}_{xl}\},\lambda_1,\lambda_2).
\end{equation*}
The domain of the dual function, {\it i.e.}, $\mathrm{dom}\,g$, is $\lambda_1,\lambda_2 \ge 0$ such that $g(\lambda_1,\lambda_2)> -\infty$. It is also called {\it dual feasible} if $(\lambda_1,\lambda_2)\in\mathrm{dom}\,g$.
It is interesting to note that $g(\lambda_1,\lambda_2)$ can be obtained by solving $L$ independent subproblems, each of which can be written as follows
\begin{equation}
\begin{aligned}
({\bf P}_1\text{-}\mathrm{sub})\; \min_{\mathbf{R}_{xl}\succeq 0} &\text{Tr}\left(\left(\mathbf{G}_2^H\mathbf{\Delta}_l\mathbf{G}_2 +\lambda_1 \mathbf{I}\right)\mathbf{R}_{xl}\right) \\
&- \lambda_2 \log_2 \left|\mathbf{I} + \mathbf{R}_{wl}^{-1}\mathbf{HR}_{xl}\mathbf{H}^H \right|.
\end{aligned}
\end{equation}
Before giving the solution of (${\bf P}_1\text{-}\mathrm{sub}$), let us first state some observations.

\noindent\emph{Observation 1)} If there is an optimal point (and it has to be unique), the average capacity constraint is active at the optimal point. This means that the achieved capacity is always $C$ and $\lambda_2>0$. To show this, let us assume that the optimal point $\{\mathbf{R}_{xl}^*\}$ achieves $\text{AC}(\{\mathbf{R}_{xl}^*\})>C$. Then we can always shrink $\{\mathbf{R}_{xl}^*\}$ until the average capacity reduces to $C$ while the objective will also be reduced. Thus, we end up with a contradiction.

\noindent\emph{Observation 2)} $\left(\mathbf{G}_2^H\mathbf{\Delta}_l\mathbf{G}_2 +\lambda_1 \mathbf{I}\right)$ is positive definite for all $l\in\mathbb{N}_L^+$. This can be shown via contradiction. Suppose that there exists $l$ such that $\mathbf{G}_2^H\mathbf{\Delta}_l\mathbf{G}_2 +\lambda_1 \mathbf{I}$ is singular. Then it must hold that $\mathbf{G}_2^H\mathbf{\Delta}_l\mathbf{G}_2$ is singular and $\lambda_1=0$. Therefore, we can always find a nonzero vector $\mathbf{v}$ lying in the null space of $\mathbf{G}_2^H\mathbf{\Delta}_l\mathbf{G}_2$. At the same time, it holds that $\mathbf{R}_{wl}^{-1/2}\mathbf{Hv}\neq 0$ with very high probability, because $\mathbf{H}$ is a realization of the random channel. If we choose $\mathbf{R}_{xl}=\alpha\mathbf{vv}^H$ and $\alpha \rightarrow \infty$, the Lagrangian $ \mathcal{L}(\{\mathbf{R}_{xl}\},0,\lambda_2)$ will be unbounded from below, which indicates that $\lambda_1=0$ is not dual feasible. This means that $\lambda_1$ is strictly larger than $0$ if $\mathbf{G}_2^H\mathbf{\Delta}_l\mathbf{G}_2$ is singular for any $l$. The claim is proved.

Based on the above observations, we have the following lemma.
\begin{lemma}[\cite{Zhang10,Kim11}] \label{lemma:closedform}
For given feasible dual variables $\lambda_1,\lambda_2 \ge 0$, the optimal solution of $({\bf P}_1\text{-}\mathrm{sub})$ is given by
\begin{equation} \label{eqn:closedform}
\mathbf{R}_{xl}^*(\lambda_1,\lambda_2) = \mathbf{\Phi}_l^{-1/2}\mathbf{U}_l\mathbf{\Sigma}_l \mathbf{U}_l^H \mathbf{\Phi}_l^{-1/2},
\end{equation}
where $\mathbf{\Phi}_l\triangleq \mathbf{G}_2^H\mathbf{\Delta}_l\mathbf{G}_2 +\lambda_1 \mathbf{I}$; $\mathbf{U}_l$ is the right singular matrix of $\tilde{\mathbf{H}}_l\triangleq \mathbf{R}_{wl}^{-1/2}\mathbf{H}\mathbf{\Phi}_l^{-1/2}$;  $\mathbf{\Sigma}_l=\mathrm{diag}(\beta_{l1},\dots,\beta_{lr})$ with $\beta_{li}=(\lambda_2-1/\sigma_{li}^2)^+$, $r$ and $\sigma_{li},i=1,\dots,r$, respectively being the rank and the positive singular vales of $\tilde{\mathbf{H}}_l$. It also holds that
\begin{equation}
\log_2 \left|\mathbf{I} + \mathbf{R}_{wl}^{-1}\mathbf{HR}_{xl}^*\mathbf{H}^H \right| =\sum_{i=1}^r\left(\log(\lambda_2\sigma_{li}^2)\right)^+ .
\end{equation}
\end{lemma}
Based on Lemma \ref{lemma:closedform}, the solution of $({\bf P}_1)$ can be obtained by finding the optimal dual variables $\lambda_1^*,\lambda_2^*$. The cooperative spectrum sharing problem $({\bf P}_1)$ can be solved via the procedure outlined in Algorithm \ref{alg:P1}.
\begin{algorithm}
\caption{Cooperative Spectrum Sharing $(\mathbf{P}_1)$}
\label{alg:P1}
\begin{algorithmic}[1]
\STATE \textbf{Input:} $\mathbf{H},\mathbf{G}_1,\mathbf{G}_2,\mathbf{\Omega}_I,P_t,C,\sigma_C^2$, $\lambda_1\in[0,\hat{\lambda}_1]$
\STATE \textbf{Initialization:} $\lambda_{l}=0,\lambda_{u}=\hat{\lambda}_1$
\REPEAT
\STATE $\lambda_1\leftarrow (\lambda_{l}+\lambda_{u})/2$
\STATE \begin{flushleft}Find the minimum~$\lambda_2~\ge~0$~such~that $\sum_{l=1}^L\sum_{i=1}^r \bigl(\log(\lambda_2\sigma_{li}^2)\bigr)^+ \ge LC$. Use the obtained $\lambda_2$ to calculate $\mathbf{R}_{xl}^*(\lambda_1,\lambda_2)$ according to (\ref{eqn:closedform}). \end{flushleft}
\IF {$\sum_{l=1}^L\text{Tr}(\mathbf{R}_{xl}^*(\lambda_1,\lambda2))< P_t$}
\STATE $\lambda_u=\lambda_1$
\ELSE
\STATE $\lambda_l=\lambda_1$
\ENDIF
\UNTIL{$\lambda_u-\lambda_l \le \delta_{\lambda}$, where $\delta_{\lambda}$ is a predefined threshold.}
\STATE $\lambda_1^*=\lambda_1,\lambda_2^*=\lambda_2$;
\STATE \textbf{Output:} $\mathbf{R}_{xl}^*=\mathbf{R}_{xl}^*(\lambda_1^*,\lambda_2^*)$
\end{algorithmic}
\end{algorithm}

Based on Lemma \ref{lemma:closedform}, the coexistence model can be equivalently viewed as a fast fading MIMO channel $\tilde{\mathbf{H}}_l$. The covariance of the waveforms transmitted on $\tilde{\mathbf{H}}_l$ is $\tilde{\mathbf{R}}_{xl}\triangleq \mathbf{\Phi}_l^{1/2}\mathbf{R}_{xl} \mathbf{\Phi}_l^{1/2}$. It is well-known that the optimum $\tilde{\mathbf{R}}_{xl}$ equals $\mathbf{U}_l\mathbf{\Sigma}_l \mathbf{U}_l^H $ with power allocation obtained by the water-filling algorithm \cite{Goldsmith03}. The achieved capacity is the average over all realization of the channel, {\it i.e.}, $\{\tilde{\mathbf{H}}_l\}_{l=1}^{L}$. This justifies the definition of average capacity in (\ref{eqn:avgcapacity}). Lemma \ref{lemma:closedform} shows that the communication transmitter will allocate more power to directions determined by the left singular vectors of $\mathbf{H}$ corresponding to larger eigenvalues and by the eigenvectors of $\mathbf{\Phi}_l$ corresponding to smaller eigenvalues. In other words, the communication will transmit more power in directions that convey larger signal at the communication receivers and smaller interferences to the MIMO-MC radars.

The following theorem compares the minimum EIP achieved by the noncooperative and cooperative approaches under the same communication constraints.
\begin{theorem} \label{thm:MIMO-MCI}
For any $P_t$ and $C$, the EIP$_{I}$ achieved by the cooperative
approaches in ($\mathbf{P}_1$) is less or equal than that of  the noncooperative approach via ($\mathbf{P}_0$).
\end{theorem}

\begin{proof}
Let $\{\mathbf{R}^{*0}_{xl}\}$  and $\{\mathbf{R}^{*1}_{xl}\}$ denote the solution
of ($\mathbf{P}_0$) and ($\mathbf{P}_1$), respectively.
We know that $\{\mathbf{R}^{*0}_{xl}\}$ satisfies the constraints in (${\bf P}_1$), which means that $\{\mathbf{R}^{*0}_{xl}\}$ is a feasible point of ($\mathbf{P}_1$). The optimal $\{\mathbf{R}^{*1}_{xl}\}$ achieves an objective value no larger than any feasible point, including $\{\mathbf{R}^{*0}_{xl}\}$, does. It holds that $\text{EIP}_{I}(\{\mathbf{R}^{*1}_{xl}\}) \le \text{EIP}_{I}(\{\mathbf{R}^{*0}_{xl}\})$, which
 proves the claim.
\end{proof}

There are certain scenarios in which the cooperative approach outperforms significantly the noncooperative one in terms of EIP. Let us denote by $\phi_1$
the intersection of $\mathcal{N}(\mathbf{G}_{2l})$ and $\mathcal{R}(\mathbf{R}_{wl}^{1/2}\mathbf{H})$, and by
$\phi_2$ the intersection of $\mathcal{N}(\mathbf{G}_{2})$ and $\mathcal{R}(\mathbf{R}_{wl}^{1/2}\mathbf{H})$. We know that $\phi_2 \subseteq \phi_1$. Consider the case where $\phi_1$ is nonempty while $\phi_2$ is empty. This happens with high probability when $M_{r,R}\ge M_{t,C}$ but $pM_{r,R}$ is much smaller than $M_{t,C}$. Problem ($\mathbf{P}_1$) will guide the communication system to focus its transmission power along the directions in $\phi_1$ to satisfy both communication system constraints, while introducing zero EIP to the radar system. On the other hand, since $\phi_2$ is empty, Problem ($\mathbf{P}_0$) will guide the communication system transmit power along directions that introduce nonzero EIP.
In other words, the sub-sampling procedure in the MIMO-MC radar may reduce the dimension of the interference channel $\mathbf{G}_2$ row space. This further increases the design flexibility of the communication waveforms. Therefore, it is more possible to find communication waveforms that satisfy the communication constraints and meanwhile introduce smaller EIP.

\subsection{Joint Communication and Radar System Design for Spectrum Sharing}\label{sec:jointdesignI}
In the above described spectrum sharing strategies, the MIMO-MC radar operates with a predetermined pseudo random sampling scheme. However, in this section, we consider a joint design of the communication system transmit covariance matrices and the MIMO-MC radar random sampling scheme, i.e., $\mathbf{\Omega}_I$. The candidate sampling scheme needs to ensure that the resulting data matrix can be completed. This means that $\mathbf{\Omega}_{I}$ is either a uniformly random sub-sampling matrix \cite{Candes10}, or a matrix with a large spectral gap \cite{Srinadh14}.

Recall that
$
\text{EIP}_{I}
=\sum\nolimits_{l=1}^L \text{Tr}\left( \mathbf{\Delta}_l\mathbf{G}_{2}\mathbf{R}_{xl}\mathbf{G}^H_{2}\right).
$
The joint design scheme is formulated as
\begin{equation*}
\begin{aligned}
({\bf P}_2)\; &\left\{\{\mathbf{R}_{xl}\},\mathbf{\Omega}_I \right\} = \argmin_{\{\mathbf{R}_{xl}\},\mathbf{\Omega}} \; \sum_{l=1}^L \text{Tr}\left( \mathbf{\Delta}_l\mathbf{G}_{2}\mathbf{R}_{xl}\mathbf{G}^H_{2}\right) \\
&\text{s.t. } \,  \{\mathbf{R}_{xl}\} \in \mathbb{X}_0,\mathbf{\Delta}_{l}=\text{diag}(\mathbf{\Omega}_{\cdot l}), \mathbf{\Omega} \text{ is proper}.
\end{aligned}
\end{equation*}
The above problem is not convex. A solution can be obtained via alternating optimization. Let $(\{\mathbf{R}_{xl}^n\},\mathbf{\Omega}^n)$ be the variables at the $n$-th iteration. We alternatively solve the following two problems:
\begin{subequations}
\begin{align}
\{\mathbf{R}_{xl}^n\} = \argmin_{\{\mathbf{R}_{xl}\}\in\mathbb{X}_0}& \sum\nolimits_{l=1}^L \text{Tr}\left( \mathbf{\Delta}_l^{n-1} \mathbf{G}_{2}\mathbf{R}_{xl}\mathbf{G}^H_{2}\right), \label{eqn:jdao1}\\
\mathbf{\Omega}^n = \argmin_{\mathbf{\Omega}}& \sum\nolimits_{l=1}^L \text{Tr}\left( \mathbf{\Delta}_l \mathbf{G}_{2}\mathbf{R}_{xl}^n \mathbf{G}^H_{2}\right), \label{eqn:jdao2}\\
\text{s.t.}\; & \mathbf{\Delta}_{l}=\text{diag}(\mathbf{\Omega}_{\cdot l}), \mathbf{\Omega} \text{ is proper}. \nonumber
\end{align}
\end{subequations}
The problem of (\ref{eqn:jdao1}) is convex and can be solved efficiently.
To avoid the intermediate variable $\{\mathbf{\Delta}_l\}$, we can reformulate (\ref{eqn:jdao2}) as
\begin{equation} \label{eqn:optomega1}
\begin{aligned}
\mathbf{\Omega} ^n = &\argmin_{\mathbf{\Omega}}\; \text{Tr}(\mathbf{\Omega}^T\mathbf{Q}^n) \; \text{s.t.} \; \mathbf{\Omega} \text{ is proper},
\end{aligned}
\end{equation}
where the $l$-th column of $\mathbf{Q}^n$ contains the diagonal entries of $\mathbf{G}_{2}\mathbf{R}_{xl}^n \mathbf{G}^H_{2}$.
Recall that the sampling matrix $\mathbf{\Omega}$ is proper either if it is a uniformly random sampling matrix, or it has large spectral gap. However, it is difficult to incorporate such conditions in the above optimization problem.

Noticing  that row and column permutation of the sampling matrix would not affect its singular values and thus the spectral gap, we propose to optimize the sampling scheme by permuting the rows and columns of an initial sampling matrix $\mathbf{\Omega}^0$:
\begin{equation} \label{eqn:optomega1omega0}
\mathbf{\Omega}^n=\argmin_{\mathbf{\Omega}}\; \text{Tr}(\mathbf{\Omega}^T \mathbf{Q}^n) \; \text{s.t.} \; \mathbf{\Omega} \in \wp(\mathbf{\Omega}^0),
\end{equation}
where $\wp(\mathbf{\Omega}^0)$ denotes the set of matrices obtained by arbitrary row and/or column permutations. The $\mathbf{\Omega}^0$ is generated with binary entries and $\lfloor p_{I}LM_{r,R} \rfloor$ ones. Meanwhile, $\mathbf{\Omega}^0$ has large spectral gap. One of the matrices that exhibit large spectral gap with high probability is the uniformly random  sampling matrix  \cite{Srinadh14}.
Brute-force search can be used to find the optimal $\mathbf{\Omega}$. However, the complexity is very high since $|\wp(\mathbf{\Omega}^0)|=\Theta(M_{r,R}!L!)$. By alternately optimizing w.r.t. row permutation and column permutation on $\mathbf{\Omega}^0$, we can solve (\ref{eqn:optomega1omega0}) using a sequence of linear assignment problems \cite{Hungarian}.

To optimize w.r.t. column permutation, we need to find the best one-to-one match between the columns of $\mathbf{\Omega}^0$ and the columns of $\mathbf{Q}^n$. We construct a cost matrix $\mathbf{C}^c\in\mathbb{R}^{L\times L}$ with $[\mathbf{C}^c]_{ml}\triangleq (\mathbf{\Omega}^0_{\cdot m})^T \mathbf{Q}^n_{\cdot l}$. The problem turns out to be a linear assignment problem with cost matrix  $\mathbf{C}^c$,  which can be solved in polynomial time using the Hungarian algorithm \cite{Hungarian}. Let  $\mathbf{\Omega}^c$ denote the column-permutated sampling matrix after the above step.
Then, we permute the rows of $\mathbf{\Omega}^c$ to optimally match the rows of $\mathbf{Q}^n$. Similarly, we construct a cost matrix $\mathbf{C}^r\in\mathbb{R}^{M_{r,R}\times M_{r,R}}$ with $[\mathbf{C}^r]_{ml}\triangleq \mathbf{\Omega}^c_{m\cdot} (\mathbf{Q}^n_{l\cdot})^T$. Again, the Hungarian algorithm can be used to solve the row assignment problem.
The above column and row permutation steps are alternately repeated until $\text{Tr}(\mathbf{\Omega}^T\mathbf{Q}^n)$ becomes smaller than a certain predefined threshold $\delta_1$.

It is easy to show that the value of $\text{EIP}_{I}$ decreases during the alternating iterations between (\ref{eqn:jdao1}) and (\ref{eqn:jdao2}). The proposed algorithm stops when $\text{EIP}_{I}$ decreases with value smaller than a certain predefined threshold $\delta_2$. The proposed joint-design spectrum sharing strategy is expected to further reduce the EIP at the Scheme I radar RX node compared to the methods in Section \ref{sec:noncooperI} and \ref{sec:cooperI}.
The complete joint-design spectrum share algorithm proposed in this section is summarized in Algorithm \ref{alg:jd1}.
\begin{algorithm}
\caption{Joint design based spectrum sharing between Scheme I radar and a MIMO comm. system}
\label{alg:jd1}
\begin{algorithmic}[1]
\STATE \textbf{Input:} $\mathbf{H},\mathbf{G}_1,\mathbf{G}_2,P_t,C,\sigma_C^2, \delta_1,\delta_2$
\STATE \textbf{Initialization:} $\mathbf{\Omega}^0$ is a uniformly random sampling matrix
\REPEAT
\STATE $\{\mathbf{R}_{xl}^n\}\leftarrow$ Solve problem (\ref{eqn:jdao1}) using {\bf Algorithm} \ref{alg:P1} while fixing $\mathbf{\Omega}^{n-1}$
\STATE $\mathbf{\Omega}^{prev}\leftarrow \mathbf{\Omega}^{n-1}$
\LOOP
\STATE $\mathbf{\Omega}^c \leftarrow$ Find the best column permutation of $\mathbf{\Omega}^{prev}$ by solving the linear assignment problem with cost matrix $\mathbf{C}^c$
\STATE $\mathbf{\Omega}^r \leftarrow$ Find the best row permutation of $\mathbf{\Omega}^{c}$ by solving the linear assignment problem with cost matrix $\mathbf{C}^r$
\IF {$|\text{Tr}((\mathbf{\Omega}^r)^T\mathbf{Q}^n) - \text{Tr}((\mathbf{\Omega}^{prev})^T\mathbf{Q}^n)| < \delta_1$}
\STATE $\textbf{Break}$
\ENDIF
\STATE $\mathbf{\Omega}^{prev}\leftarrow \mathbf{\Omega}^{r}$
\ENDLOOP
\STATE $\mathbf{\Omega}^n \leftarrow \mathbf{\Omega}^r$
\STATE $n \leftarrow n+1$
\UNTIL{$|\text{EIP}_{I}^n-\text{EIP}_{I}^{n-1}| < \delta_2$ }
\STATE \textbf{Output:} $\{\mathbf{R}_{xl}\} = \{\mathbf{R}_{xl}^n\},\mathbf{\Omega}_{I}=\mathbf{\Omega}^n$
\end{algorithmic}
\end{algorithm}

\section{Spectrum Sharing with Scheme II MIMO-MC Radars} \label{sec:schemeII}
When the Scheme II radar is considered, the signal model of the random matched filter can be expressed as follows:
\begin{equation*}
\mathbf{\Omega}_{II}\circ (\mathbf{Y}_R\mathbf{S}^H)={\mathbf{\Omega}_{II}}\circ (\mathbf{D}\mathbf{S}\mathbf{S}^H+\mathbf{G}_2\mathbf{X}\mathbf{\Lambda}_2\mathbf{S}^H+\mathbf{W}_R\mathbf{S}^H).
\end{equation*}
\newcounter{MYtempeqncnt}
\begin{figure*}[!t]
\normalsize
\begin{equation}\label{eqn:interfII}
\begin{aligned}
\text{EIP}_{II}\triangleq \mathbb{E}\left\{\text{Tr}\left(\mathbf{\Omega}_{II}\circ (\mathbf{G}_2\mathbf{X}\mathbf{\Lambda}_2\mathbf{S}^H) \left(\mathbf{\Omega}_{II} \circ ( \mathbf{G}_2\mathbf{X}\mathbf{\Lambda}_2\mathbf{S}^H) \right)^H\right)\right\}
=\mathbb{E}\left\{\sum_{m=1}^{M_{r,R}} \mathbf{g}_m^H\mathbf{X}\mathbf{\Lambda}_2\mathbf{S}_m^H \mathbf{S}_m\mathbf{\Lambda}_2^H\mathbf{X}^H \mathbf{g}_m \right\}
\end{aligned}
\end{equation}
\hrulefill
\vspace*{4pt}
\end{figure*}
The effective interference power to the Scheme II radar is given by (\ref{eqn:interfII}) on top of next page, where $\mathbf{g}_m^H$ denotes the $m$-th row of $\mathbf{G}_2$; $\mathbf{S}_m$ is composed by rows selected from $\mathbf{S}$ according to set $\xi_m$ as defined in Section \ref{sec:background}. Each sum term on the right hand side (RHS) of (\ref{eqn:interfII}) is the interference power at one radar receive antenna.
To minimize the interference power with respect to the spatial spectrum $\{\mathbf{R}_{xl}\}$, we have the following lemma to express (\ref{eqn:interfII}) in terms of $\{\mathbf{R}_{xl}\}$.
\begin{lemma} \label{lemma:EIPII1}
For the effective interference power $\text{EIP}_{II}$, it holds that
\begin{equation} \label{eqn:EIPII1}
\text{EIP}_{II} = \sum_{l=1}^L \text{Tr}\left( \mathbf{\Delta}_{l\xi}  \mathbf{G}_2 \mathbf{R}_{xl} \mathbf{G}^H_2 \right),
\end{equation}
where $ \mathbf{\Delta}_{l\xi} \triangleq \text{diag}(a_{l \xi_1 },\dots,a_{l \xi_{M_{r,R}} } )$; $a_{l \xi_m} =\mathbf{s}_m^H(l)\mathbf{s}_m(l)$ with $\mathbf{s}_m(l)$ containing entries of $\mathbf{s}(l)$ indexed by set $\xi_m$.
\end{lemma}
\begin{IEEEproof}
The proof can be found in Appendix \ref{appen:EIPII1}.
\end{IEEEproof}

If we choose $\xi_m=\mathbb{N}_{M_{t,R}}^+$, {\it i.e.}, all matched filters are used and no matrix completion is considered, $\mathbf{S}_m$ equals $\mathbf{S}$ for all $m\in\mathbb{N}_{M_{r,R}}^+$. Then by Lemma \ref{lemma:EIPII1}, the interference at the output of the full matched filter bank equals
\begin{equation} \label{eqn:IPFMFB}
\text{IP}_{\text{FMFB}}\triangleq \sum_{l=1}^{L}a_l\text{Tr}\left( \mathbf{G}_2 \mathbf{R}_{xl} \mathbf{G}^H_2 \right),
\end{equation}
where $a_{l} \triangleq \mathbf{s}^H(l) \mathbf{s}(l)$. It is noted that $0< a_{l \xi_m} < a_{l}, \forall m\in\mathbb{N}_{M_{r,R}}^+$.

In the following we discuss four levels of cooperation between the communication system and the Scheme II radar.
\subsection{Noncooperative Spectrum Sharing}
In the first case, the communication transmitter does not utilize any knowledge of the MIMO radar system except for the interference channel $\mathbf{G}_2$. Just as in the noncooperative case in Section \ref{sec:schemeI}, the communication transmitter designs its spectrum to minimize the interference power exerted at the radar RX antennas, {\it i.e.}, $\text{TIP}$, using ($\mathbf{P}_0$).
\subsection{Partially Cooperative Spectrum Sharing}
In the second case, the communication transmitter exploits knowledge of the $a_l$'s, obtained by using shared radar waveforms\footnote{Recall that the communication capacity in (\ref{eqn:avgcapacity}) is defined based on the knowledge of $\mathbf{S}$. This means that the radar waveforms are shared with the communication transmitter.}.
The communication transmitter designs its spectrum to minimize the interference power at the output of full matched filter banks in all the radar receivers
\begin{equation}
\begin{aligned}
({\bf P}_3)\; \min_{\{\mathbf{R}_{xl}\}} \; \text{IP}_{\text{FMFB}}(\{\mathbf{R}_{xl}\})   
\;\text{s.t. } &\; \{\mathbf{R}_{xl}\} \in \mathbb{X}_0.
\end{aligned}
\end{equation}
The interference power $\text{IP}_{\text{FMFB}}$ has the same summation terms as in $\text{TIP}$ but reweighed by the $a_l$'s along different symbol durations.
\subsection{Fully Cooperative Spectrum Sharing}
In the fully cooperative case, the radar system shares the diagonal matrices $\mathbf{\Delta}_{l\xi}, l\in\mathbb{N}_L^+$ with the communication system. The spectrum sharing problem can be formulated as
\begin{equation}
\begin{aligned}
({\bf P}_4)\; \min_{\{\mathbf{R}_{xl}\}} \; \text{EIP}_{II}(\{\mathbf{R}_{xl}\})   
\;\text{s.t. } &\; \{\mathbf{R}_{xl}\} \in \mathbb{X}_0.
\end{aligned}
\end{equation}
The effective interference power $\text{EIP}_{II}$ also has similar structure to $\text{TIP}$ and $\text{IP}_{\text{FMFB}}$ in (\ref{eqn:interfarrive}) and (\ref{eqn:IPFMFB}), respectively, while it is reweighed by the diagonal matrices $\mathbf{\Delta}_{l\xi}, l\in\mathbb{N}_L^+$.
We can see that the random matched filter bank introduces the weights $\mathbf{\Delta}_{l\xi}$'s which affect the power allocation in both time and spatial domain. (${\bf P}_4$) can also be solved using the dual decomposition technique used in Algorithm \ref{alg:P1}.

The following theorem compares the effective interference power to Scheme II radar, achieved by ($\mathbf{P}_0$), ($\mathbf{P}_3$), ($\mathbf{P}_4$) in the above three cases.
\begin{theorem}\label{thm:MIMO-MCII}
For any $P_t$ and $C$, the effective interference power to Scheme II radar achieved by ($\mathbf{P}_4$) is not larger than those achieved by ($\mathbf{P}_0$) and ($\mathbf{P}_3$) when none or partial information is shared with the communication transmitter.
\end{theorem}
\begin{IEEEproof}
Let $\{\mathbf{R}^{*0}_{xl}\}$, $\{\mathbf{R}^{*3}_{xl}\}$ and $\{\mathbf{R}^{*4}_{xl}\}$ denote the solution of ($\mathbf{P}_0$), ($\mathbf{P}_3$) and ($\mathbf{P}_4$), respectively.
We know that both $\{\mathbf{R}^{*0}_{xl}\}$ and $\{\mathbf{R}^{*3}_{xl}\}$ satisfy the constraints in (${\bf P}_4$), which means that $\{\mathbf{R}^{*0}_{xl}\}$ and $\{\mathbf{R}^{*3}_{xl}\}$ are two feasible points of ($\mathbf{P}_4$). Meanwhile, the optimal $\{\mathbf{R}^{*4}_{xl}\}$ achieves an objective value no larger than any feasible point, including $\{\mathbf{R}^{*0}_{xl}\}$ and $\{\mathbf{R}^{*3}_{xl}\}$, does. It holds that $\text{EIP}_{II}(\{\mathbf{R}^{*4}_{xl}\}) \le \text{EIP}_{II}(\{\mathbf{R}^{*0}_{xl}\})$ and $\text{EIP}_{II}(\{\mathbf{R}^{*4}_{xl}\}) \le \text{EIP}_{II}(\{\mathbf{R}^{*3}_{xl}\})$, which
 prove the claim.
\end{IEEEproof}

\subsection{Joint Communication and Radar System Design for Spectrum Sharing}\label{sec:jointdesignII}
In the above described spectrum sharing strategies, the Scheme II radar operates with a predetermined pseudo random sampling scheme.
In this section, we consider a joint design of the communication system transmit covariance matrices and the MIMO-MC radar sampling scheme, i.e., $\mathbf{\Omega}_{II}$.
The key of applying the joint design scheme is to express $\mathbf{\Delta}_{l\xi}$ in terms of $\mathbf{\Omega}_{II}$, which is given in the following lemma.
\begin{lemma} \label{lemma:EIPII2}
The effective interference power $\text{EIP}_{II}$ can be equivalently expressed as
\begin{equation}
\text{EIP}_{II} = \text{Tr}\left(\mathbf{\Omega}_{II}^T\mathbf{Q}(\mathbf{S}\circ\mathbf{S})^T \right),
\end{equation}
where the $l$-th column of $\mathbf{Q}$ contains the diagonal entries of $\mathbf{G}_{2}\mathbf{R}_{xl} \mathbf{G}^H_{2}$.
\end{lemma}

\begin{IEEEproof}
The proof can be found in Appendix \ref{appen:EIPII2}.
\end{IEEEproof}
The joint design scheme is formulated as follows
\begin{equation*}
\begin{aligned}
({\bf P}_5)\; \left\{\{\mathbf{R}_{xl}\},\mathbf{\Omega}_{II}\right\}=&\argmin_{\{\mathbf{R}_{xl}\},\mathbf{\Omega}} \; \sum_{l=1}^L \text{Tr}\left( \mathbf{\Delta}_{l\xi}\mathbf{G}_{2}\mathbf{R}_{xl}\mathbf{G}^H_{2}\right) \\
&\text{s.t. } \,  \{\mathbf{R}_{xl}\} \in \mathbb{X}_0,\mathbf{\Omega} \in \wp(\mathbf{\Omega}^0).
\end{aligned}
\end{equation*}
As in problem $({\bf P}_2)$, a suboptimal sampling matrix $\mathbf{\Omega}$ is searched over the set of matrices obtained by permutating rows and/or columns of $\mathbf{\Omega}^0$. The initial $M_{r,R}\times M_{t,R}$ dimensional matrix $\mathbf{\Omega}^0$ is generated with $\lfloor p_{II}M_{t,R}M_{r,R} \rfloor$ ones at uniformly random positions. This guarantees that $\mathbf{\Omega}^0$ and matrices obtained by permutating rows and/or columns of $\mathbf{\Omega}^0$ have large spectral gap. Multiple instances of $\mathbf{\Omega}^0$ can be used to find a better radar sampling scheme.
Similarly, the technique of alternating optimization is adopted to solve $({\bf P}_5)$. Let $(\{\mathbf{R}_{xl}^n\},\mathbf{\Omega}^n)$ be the variables at the $n$-th iteration. We alternatively solve the following two problems:
\begin{small}
\begin{subequations}
\begin{align}
\{\mathbf{R}_{xl}^n\} = \argmin_{\{\mathbf{R}_{xl}\}\in\mathbb{X}_0}& \sum_{l=1}^L \text{Tr}\left( \text{diag}\left(\mathbf{\Omega}^{n-1} \left(\mathbf{s}(l)\circ \mathbf{s}(l)\right)\right) \mathbf{G}_{2}\mathbf{R}_{xl}\mathbf{G}^H_{2}\right), \label{eqn:jdaoII1}\\
\mathbf{\Omega}^n = \argmin_{\mathbf{\Omega}}& \, \text{Tr}\left(\mathbf{\Omega}^T\mathbf{Q}^n(\mathbf{S}\circ\mathbf{S})^T \right), \text{s.t.}\; \mathbf{\Omega} \in \wp(\mathbf{\Omega}^{n-1}), \label{eqn:jdaoII2} 
\end{align}
\end{subequations}
\end{small}
where the $l$-th column of $\mathbf{Q}^n$ contains the diagonal entries of $\mathbf{G}_{2}\mathbf{R}_{xl}^n \mathbf{G}^H_{2}$.
The problem of (\ref{eqn:jdaoII1}) is convex and can be solved efficiently using Algorithm \ref{alg:P1}. By denoting $\tilde{\mathbf{Q}}^n\triangleq \mathbf{Q}^n(\mathbf{S}\circ\mathbf{S})^T$, subproblem (\ref{eqn:jdaoII2}) can be formulated into exactly the same form as (\ref{eqn:optomega1omega0}). Analogously, (\ref{eqn:jdaoII2}) is solved using a sequence of linear assignment problems \cite{Hungarian}, which alternately optimize w.r.t. row permutation and column permutation on $\mathbf{\Omega}^0$. The corresponding cost matrices $\mathbf{C}^c$ and $\mathbf{C}^r$ are with entries given by $[\mathbf{C}^c]_{ml}\triangleq (\mathbf{\Omega}^0_{\cdot m})^T \tilde{\mathbf{Q}}^n_{\cdot l}$ and $[\mathbf{C}^r]_{ml}\triangleq \mathbf{\Omega}^c_{m\cdot} (\tilde{\mathbf{Q}}^n_{l\cdot})^T$, respectively.
The complete joint-design based spectrum sharing algorithm proposed in this section is summarized in Algorithm \ref{alg:jd2}.
\begin{algorithm}
\caption{Joint design based spectrum sharing between Scheme II radar and a MIMO comm. system}
\label{alg:jd2}
\begin{algorithmic}[1]
\STATE \textbf{Input:} $\mathbf{H},\mathbf{G}_1,\mathbf{G}_2,\mathbf{S},P_t,C,\sigma_C^2,\delta_1,\delta_2$
\STATE \textbf{Initialization:} $\mathbf{\Omega}^0$ is a uniformly random sampling matrix
\REPEAT
\STATE $\{\mathbf{R}_{xl}^n\}\leftarrow$ Solve problem (\ref{eqn:jdaoII1}) using {\bf Algorithm} \ref{alg:P1} while fixing $\mathbf{\Omega}^{n-1}$
\STATE $\mathbf{\Omega}^{prev}\leftarrow \mathbf{\Omega}^{n-1}$
\LOOP
\STATE $\mathbf{\Omega}^c \leftarrow$ Find the best column permutation of $\mathbf{\Omega}^{prev}$ by solving the linear assignment problem with cost matrix $\mathbf{C}^c$
\STATE $\mathbf{\Omega}^r \leftarrow$ Find the best row permutation of $\mathbf{\Omega}^{c}$ by solving the linear assignment problem with cost matrix $\mathbf{C}^r$
\IF {$|\text{Tr}((\mathbf{\Omega}^r)^T\mathbf{Q}^n(\mathbf{S}\circ\mathbf{S})^T) - \text{Tr}((\mathbf{\Omega}^{prev})^T\mathbf{Q}^n(\mathbf{S}\circ\mathbf{S})^T)| < \delta_1$}
\STATE $\textbf{Break}$
\ENDIF
\STATE $\mathbf{\Omega}^{prev}\leftarrow \mathbf{\Omega}^{r}$
\ENDLOOP
\STATE $\mathbf{\Omega}^n \leftarrow \mathbf{\Omega}^r$
\STATE $n \leftarrow n+1$
\UNTIL{$|\text{EIP}_{II}^n-\text{EIP}_{II}^{n-1}| < \delta_2$ }
\STATE \textbf{Output:} $\{\mathbf{R}_{xl}\} = \{\mathbf{R}_{xl}^n\},\mathbf{\Omega}_{II}=\mathbf{\Omega}^n$
\end{algorithmic}
\end{algorithm}

\section{Spectrum Sharing between Mismatched Systems} \label{sec:discussions}
In Section \ref{sec:sigmodel}, the waveform symbol duration of the radar system is assumed to match that of the communication system. For a typical communication channel with $40\times 10^6Hz$ bandwidth, the maximum symbol rate is $20\times 10^6$ symbols/$s$. For our assumption to be valid, the radar waveform symbol duration need to be $\frac{1}{20} \mu s$, which results in a typical range resolution of $7.5$ meters.

In the following, we consider the mismatched cases. We will show that the proposed techniques presented in the previous sections can still be applied.
Let $f_{s}^R=1/T_R$ and $f_{s}^C$ denote the radar waveform symbol rate and the communication symbol rate, respectively.
Also, let the length of radar waveforms be denoted by $L_R$. The number of communication symbols transmitted in the duration of $L_R/f_s^R$ is $L_C \triangleq \lceil {L_R f_s^C f_s^R}\rceil$. The communication average capacity and transmit power can be expressed in terms of $\{\mathbf{R}_{xl}\}_{l=1}^{L_C}$ as in Section \ref{sec:schemeI}. In the following, we will only focus on the effective interference to the MIMO-MC radar receiver.

If $f_s^R<f_s^C$, the interference arrived at the radar receiver will be down-sampled. Let $\mathcal{I}_1\subset\mathbb{N}^+_{L_C}$ be the set of indices of communication symbols that are sampled by the radar in ascending order. It holds that $|\mathcal{I}_1|=L_R$. Following the derivation in previous sections, we have the following interference power expressions:
\begin{equation*}
\begin{aligned}
\text{EIP}_{I}&=\sum\nolimits_{l\in\mathcal{I}_1}\text{Tr}\left( \mathbf{\Delta}_{l'} \mathbf{G}_{2}\mathbf{R}_{xl}\mathbf{G}^H_{2} \right), \\
\text{EIP}_{II}&=\sum\nolimits_{l\in\mathcal{I}_1}\text{Tr}\left( \mathbf{\Delta}_{l'\xi} \mathbf{G}_{2}\mathbf{R}_{xl}\mathbf{G}^H_{2} \right),
\end{aligned}
\end{equation*}
where $l'\in\mathbb{N}^+_{L_R}$ is the index of $l$ in ordered set $\mathcal{I}_1$. We observe that the communication symbols indexed by $\mathbb{N}^+_{L_C}\setminus\mathcal{I}_1$, which are not sampled by the radar receiver, would introduce zero interference power to the radar system.

If $f_s^R> f_s^C$, the interference arrived at the radar receiver will be over-sampled.
One individual communication symbol will introduce interference to the radar system in $\lfloor f_s^R/f_s^C \rfloor$ consecutive symbol durations.
Let $\tilde{\mathcal{I}}_l$ be the set of radar sampling time instances during the period of the $l$-th communication symbol. Note that $\tilde{\mathcal{I}}_l$ is with cardinality $\lfloor f_s^R/f_s^C \rfloor$, and the collection of sets $\tilde{\mathcal{I}}_1,\dots,\tilde{\mathcal{I}}_{L_C}$ is a partition of $\mathbb{N}_{L_R}^+$. The effective interference power for both schemes of MIMO-MC radar is respectively
\begin{equation*}
\begin{aligned}
\text{EIP}_{I}&=\sum\nolimits_{l=1}^{L_C}\text{Tr}\left( \tilde{\mathbf{\Delta}}_l\mathbf{G}_{2}\mathbf{R}_{xl}\mathbf{G}^H_{2} \right),\\
\text{EIP}_{II}&=\sum\nolimits_{l=1}^{L_C}\text{Tr}\left( \tilde{\mathbf{\Delta}}_{l\xi}\mathbf{G}_{2}\mathbf{R}_{xl}\mathbf{G}^H_{2} \right),
\end{aligned}
\end{equation*}
where $\tilde{\mathbf{\Delta}}_l = \sum\nolimits_{l'\in\tilde{\mathcal{I}}_{l}}\mathbf{\Delta}_{l'}$ and $\tilde{\mathbf{\Delta}}_{l\xi} = \sum\nolimits_{l'\in\tilde{\mathcal{I}}_{l}}\mathbf{\Delta}_{l'\xi}$.
We observe that each individual communication transmit covariance matrix will be weighted by the sum of interference channels for $\lfloor f_s^R/f_s^C \rfloor$ radar symbol durations instead of one single interference channel.

We conclude that in the above mismatched cases, the EIP expressions have the same form as those in the matched case except the diagonal matrices $\mathbf{\Delta}_l$ and $\mathbf{\Delta}_{l\xi}$. To calculate the corresponding diagonal matrices, the communication system only needs to know the sampling time of the radar system. Therefore, the spectrum sharing problems in such cases can still be solved using the proposed algorithms of Sections \ref{sec:schemeI} and \ref{sec:schemeII}. Further investigation will be considered as our future work.

\section{Numerical Results} \label{sec:simulation}
For the simulations, we set the number of symbols to $L=32$ and the noise variance to $\sigma^2_C=0.01$.
The MIMO radar system consists of colocated TX and RX antennas forming half-wavelength uniform linear arrays, and transmitting Gaussian orthogonal waveforms \cite{Sun13}.
The channel $\mathbf{H}$ is taken to have independent entries, distributed as $\mathcal{CN}(0,1)$. The interference channels $\mathbf{G}_1$ and $\mathbf{G}_2$ are generated with independent entries, distributed as $\mathcal{CN}(0,\sigma_1^2)$ and $\mathcal{CN}(0,\sigma_2^2)$, respectively. We fix $\sigma_1^2=\sigma_2^2=0.1$ unless otherwise stated. The maximum communication transmit power is set to $P_t = L$ (the power is normalized w.r.t the power of radar waveforms). The propagation path from the radar TX antennas to the radar RX antennas via the far-field target introduces a much more severe loss of power, $\gamma^2$, which is set to $-30$dB in the simulations.
The transmit power of the radar antennas is fixed to $\rho^2=\rho_0\triangleq1000L/M_{t,R}$ unless otherwise stated, and noise in the received signal is added at SNR$=25$dB. The phase jitter variance is taken to be $\sigma_{\alpha}^2=10^{-3}$.
The same uniformly random sampling scheme $\mathbf{\Omega}^0$ is adopted by the radar in both the noncooperative and the cooperative spectrum sharing (SS) methods. The joint-design SS method uses the same sampling matrix as its initial sampling matrix.
The TFOCUS package \cite{TFOCUS} is used for low-rank matrix completion at the radar fusion center.
The communication covariance matrix is optimized according to the criteria of Sections \ref{sec:schemeI} and \ref{sec:schemeII}. The obtained $\mathbf{R}_{xl}$ is used to generate $\mathbf{x}(l)=\mathbf{R}_{xl}^{1/2}\text{randn}(M_{t,C},1)$. We use the EIP and MC relative recovery error as the performance metrics. The relative recovery error is defined as $\lVert \mathbf{DS}-\hat{\mathbf{DS}}\rVert_F /\lVert \mathbf{DS}\rVert_F$ for Scheme I and $\lVert \mathbf{D}-\hat{\mathbf{D}}\rVert_F/\lVert \mathbf{D}\rVert_F$ for Scheme II, where $\hat{\mathbf{DS}}$ and $\hat{\mathbf{D}}$ are the completed results of $\mathbf{DS}$ and $\mathbf{D}$, respectively.
For comparison, we also implement a ``selfish communication" scenario,
where the communication system minimizes the transmit power to achieve certain average capacity  without any concern about the interferences it exerts to the radar system.

\subsection{Spectrum Sharing between a Scheme I radar and a MIMO Communication System }
\begin{figure}[htb]
\vspace{-2mm}
  \centering
  \subfigure{
  \includegraphics[width=4.3cm]{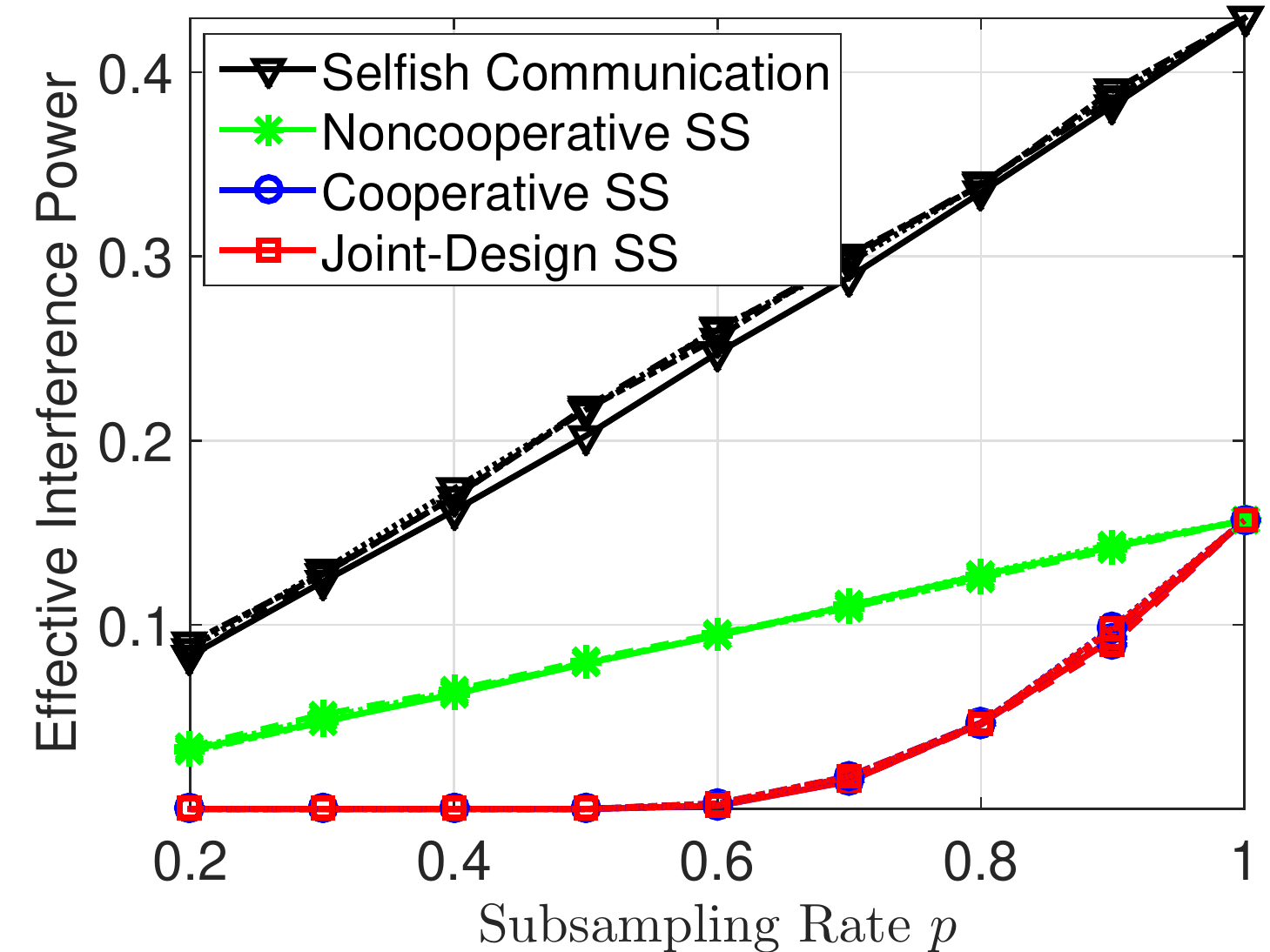}
  }
  \hspace{-5mm}
  \subfigure{
  \includegraphics[width=4.3cm]{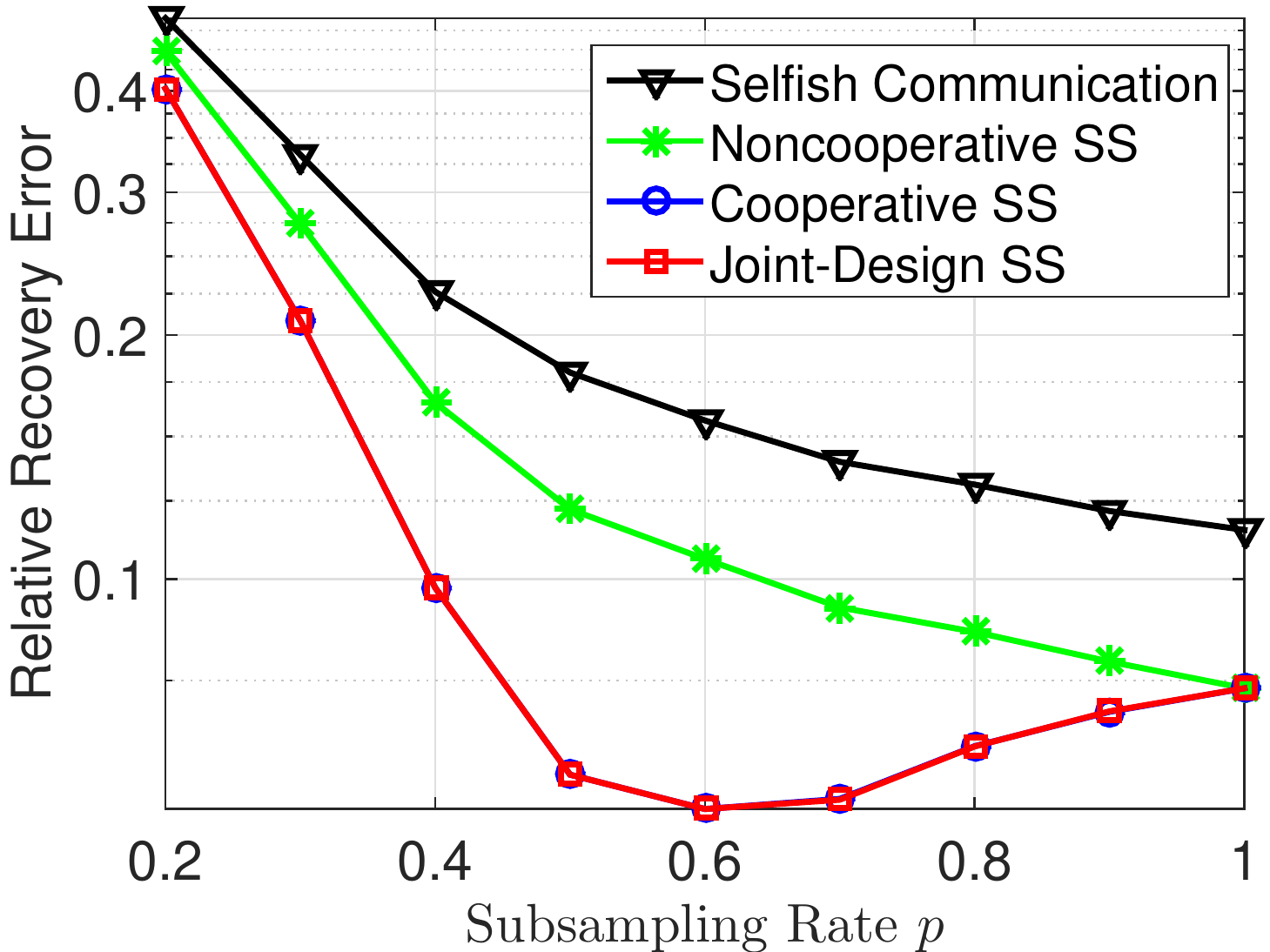}
  }
  \vspace{-5mm}
  \caption{Spectrum sharing with the Scheme I radar under different sub-sampling rates. $M_{t,R}=4, M_{r,R}=M_{t,C}=8,M_{r,C}=4$.} \label{fig:SchemeIpercent1}
\end{figure}

\begin{figure}[htb]
\vspace{-2mm}
  \centering
  \subfigure{
  \includegraphics[width=4.3cm]{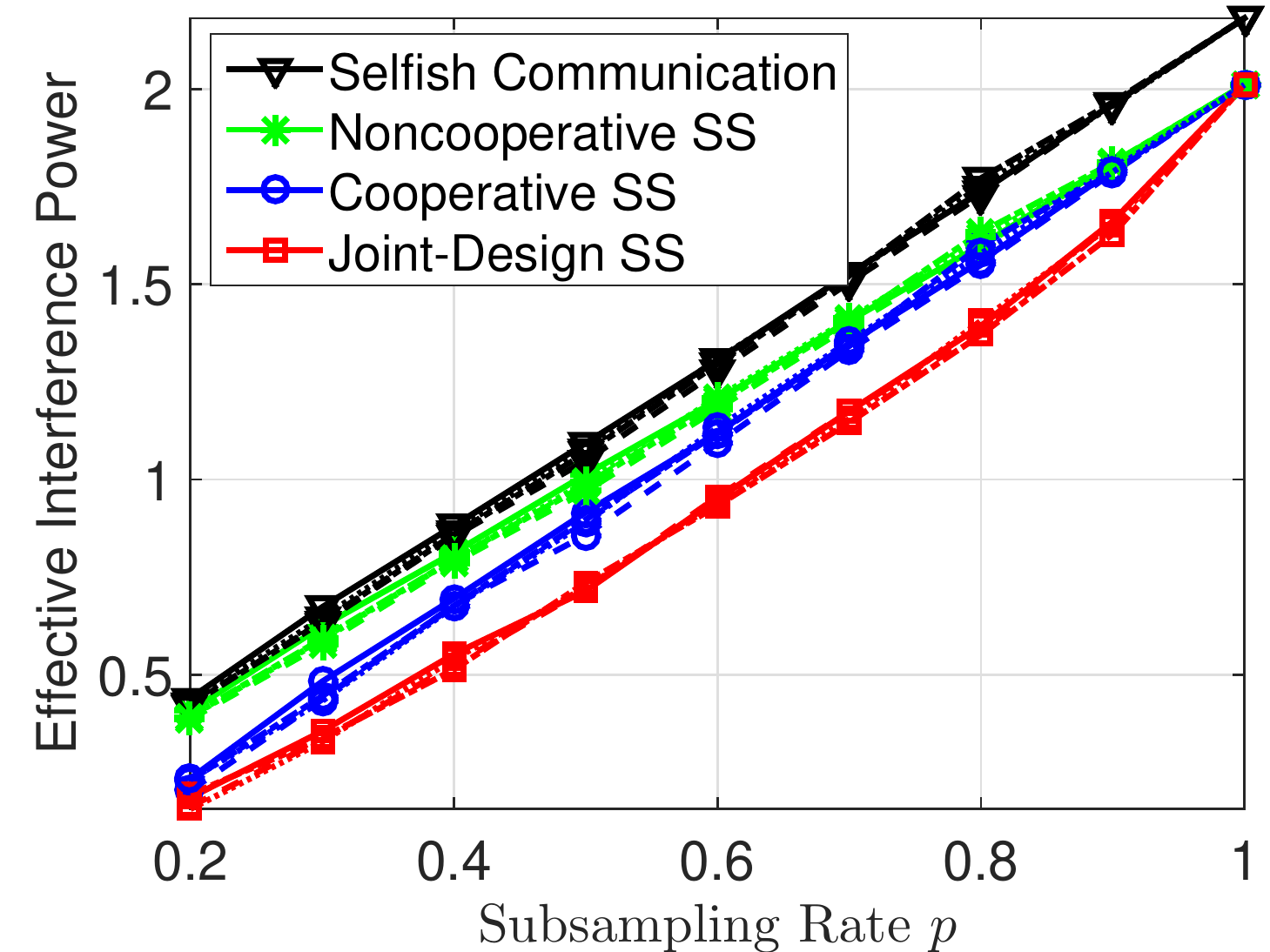}
  }
  \hspace{-5mm}
  \subfigure{
  \includegraphics[width=4.3cm]{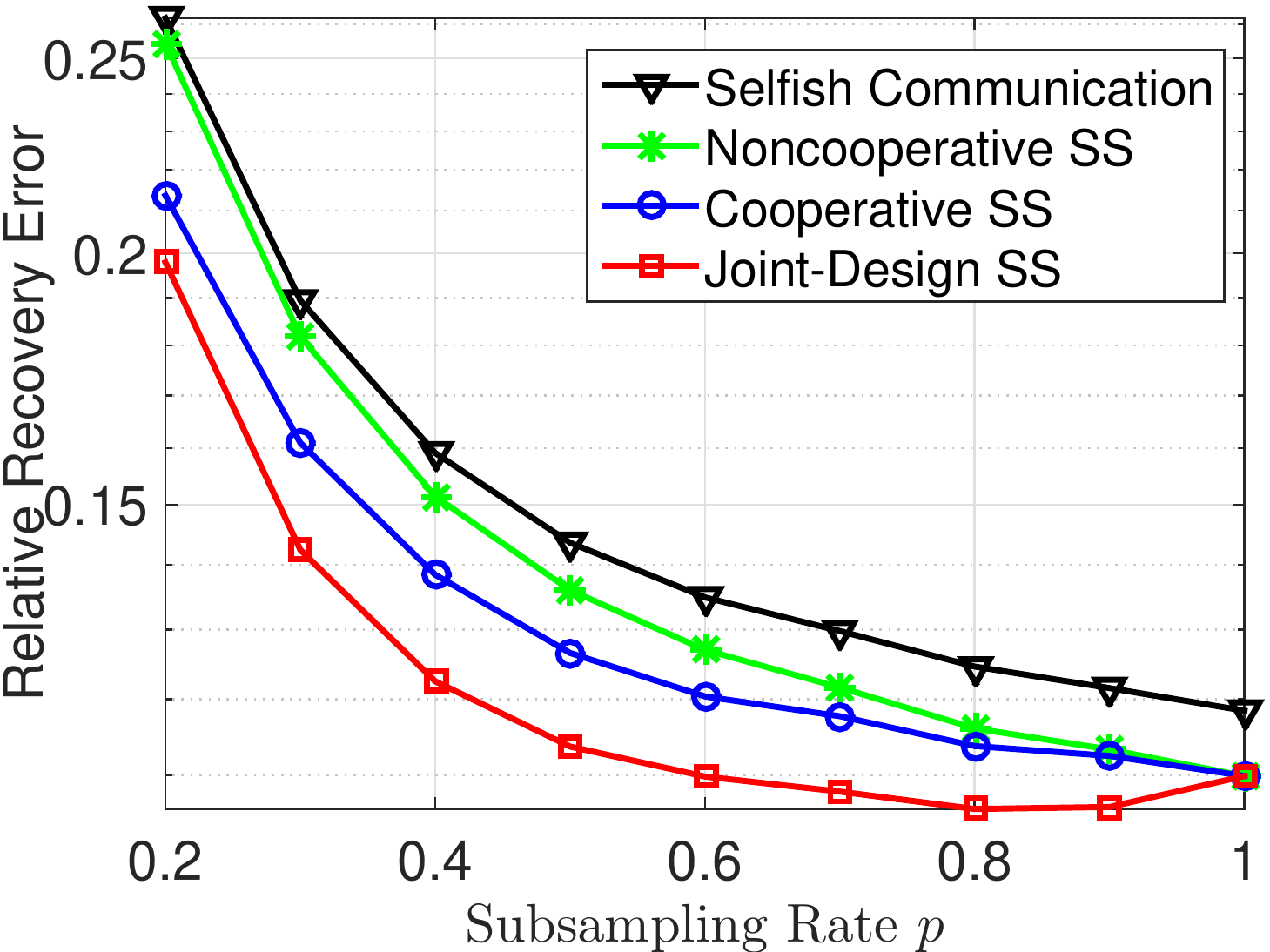}
  }
  \vspace{-5mm}
  \caption{Spectrum sharing with the Scheme I radar under different sub-sampling rates. $M_{t,R}=16,M_{r,R}=32, M_{t,C}=M_{r,C}=4$. } \label{fig:SchemeIpercent2}
\end{figure}

\subsubsection{Performance under different sub-sampling rates}
There is a far-field stationary target at angle $30^\circ$ w.r.t. the radar arrays, with target reflection coefficient equal to $0.2+0.1j$.
For the communication capacity constraint, we consider $C = 12$ bits/symbol. 
The sub-sampling rate of Scheme I radar varies from $0.2$ to $1$. The following two scenarios are considered.

In the first scenario, we use $M_{t,R}=4, M_{r,R}=M_{t,C}=8,M_{r,C}=4$. We plot the EIP results for $4$ different realizations of $\mathbf{\Omega}^0$ in Fig. \ref{fig:SchemeIpercent1}(a). For better visualization, Fig. \ref{fig:SchemeIpercent1}(b) shows the relative recovery errors averaged over all $4$ realization of $\mathbf{\Omega}^0$.
The cooperative spectrum sharing (SS) method (see ($\mathbf{P}_1$)) outperforms its noncooperative counterpart (see ($\mathbf{P}_0$)) in terms of both EIP and MC relative recovery error.  As discussed in Section \ref{sec:schemeI}, the EIP is significantly reduced by the cooperative SS method when $p < 0.6$, i.e., when $pM_{r,R}$ is much smaller than $M_{t,C}$.
The cooperative SS method performs almost the same as the joint-design method in this scenario. One possible reason is that the row dimension of $\mathbf{\Omega}$ is too small to generate sufficient difference in EIP among the permutations of $\mathbf{\Omega}$.

In the second scenario,
we choose $M_{t,R}=16,M_{r,R}=32, M_{t,C}=4,M_{r,C}=4$. In Fig. \ref{fig:SchemeIpercent2}(a), we plot the EIP corresponding to $4$ different realization of $\mathbf{\Omega}^0$. Again, Fig. \ref{fig:SchemeIpercent2}(b) shows the relative recovery errors averaged over all $4$ realization of $\mathbf{\Omega}^0$. The cooperative SS method outperforms the noncooperative SS one only marginally. This is due to the fact that both $\mathbf{G}_2$ and $\mathbf{G}_{2l}$ are full rank. The joint-design method for SS in Section \ref{sec:jointdesignI} optimizes $\mathbf{\Omega}$ starting from the same sampling matrix used by the other three methods.
Fig. \ref{fig:SchemeIpercent2} suggests that the joint-design SS method achieves smaller EIP and relative recovery errors than the other three methods.

We should note that when $p$ decreases, the null space of $\mathbf{G}_{2l}$ expands with high probability, and the EIP of the cooperative SS method is reduced. However, if $p$ is too small, the MC recovery at the fusion center fails. In the above scenarios, we would like $p\ge 0.4$ for a small relative recovery error in matrix completion. However, values of $p>0.6$ require more samples while achieving little, or even no improvement on the relative recovery error. Therefore, the optimal range of $p$ is $[0.4,0.6]$, where the proposed joint-design SS method reduces the EIP by at least $20$\% over the ``selfish communication method".
In conclusion, the sub-sampling procedure in Scheme I radar is beneficial in terms of reducing the effective interference power from the communication system as well as reducing the amount of data to be sent to the fusion center. In addition, simulations indicate that the communication average capacity constraint holds with equality in both scenarios, confirming observation $(1)$ of Section \ref{sec:cooperI}.

\begin{figure}[htb]
\vspace{-2mm}
  \centering
  \subfigure{
  \includegraphics[width=4.3cm]{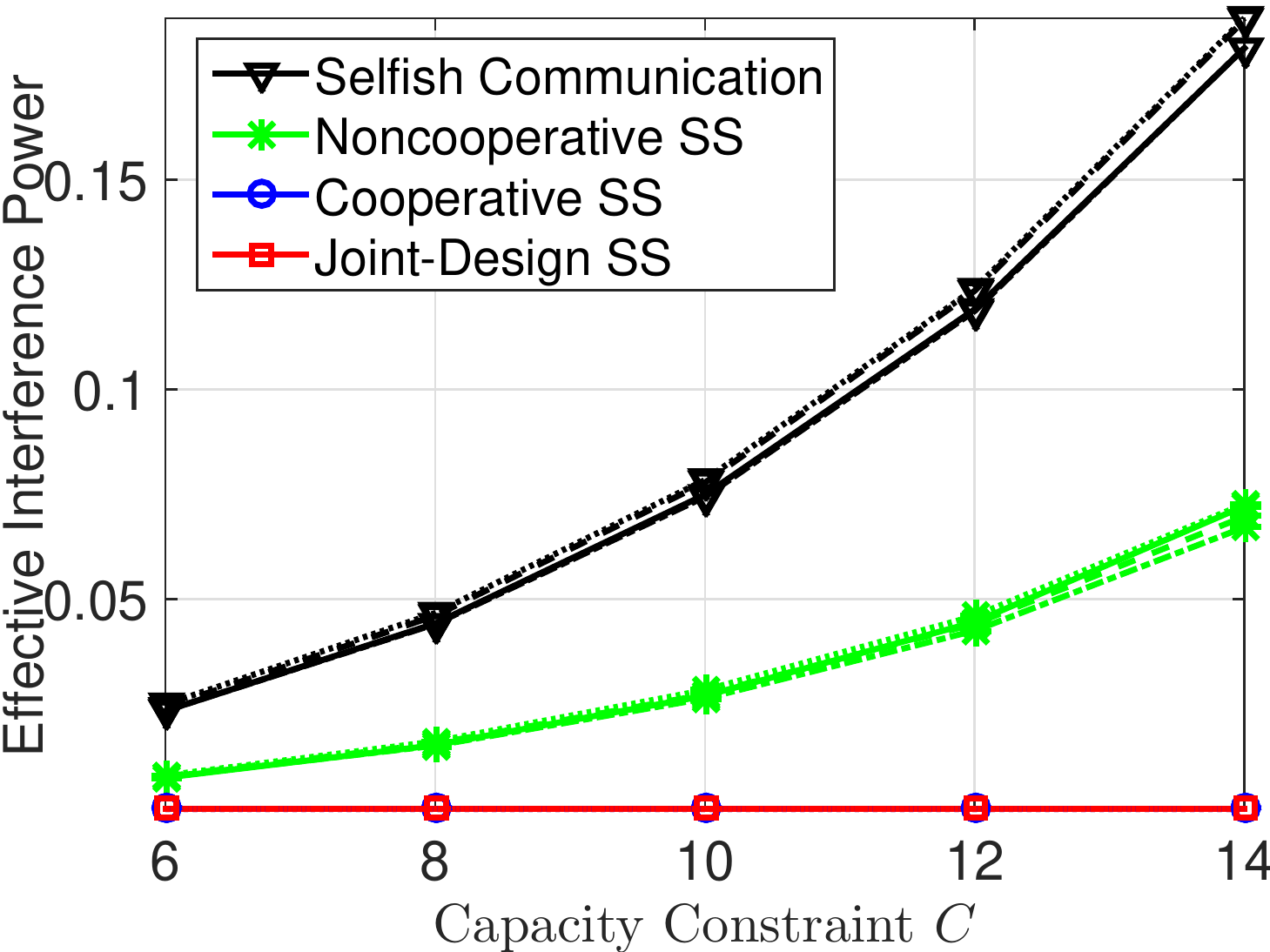}
  }
  \hspace{-5mm}
  \subfigure{
  \includegraphics[width=4.3cm]{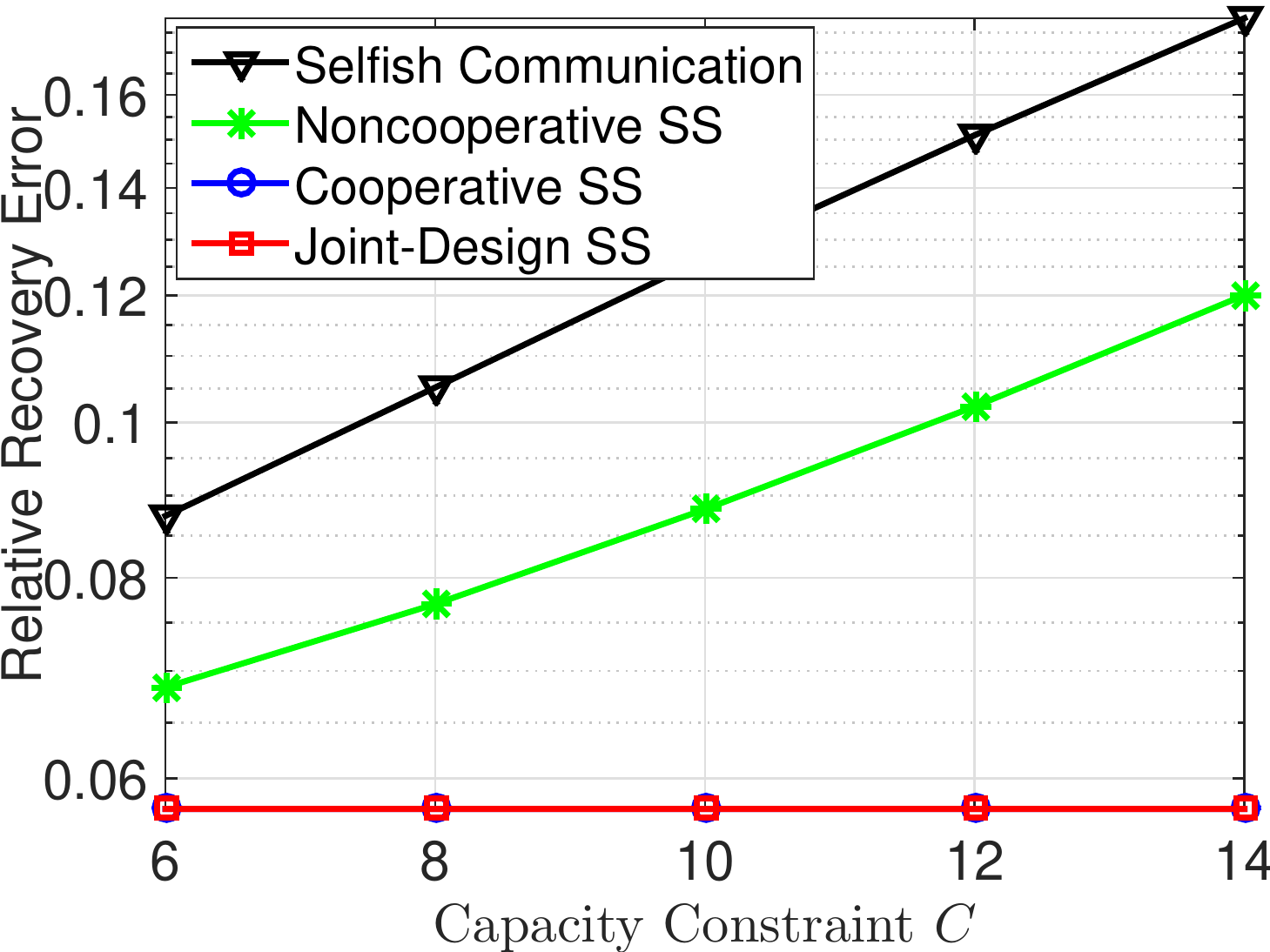}
  }
  \vspace{-5mm}
  \caption{Spectrum sharing with the Scheme I radar under different capacity constraints $C$. $M_{t,R}=4, M_{r,R}=M_{t,C}=8,M_{r,C}=4$.} \label{fig:SchemeIcapacity1}
\end{figure}

\begin{figure}[htb]
\vspace{-2mm}
  \centering
  \subfigure{
  \includegraphics[width=4.3cm]{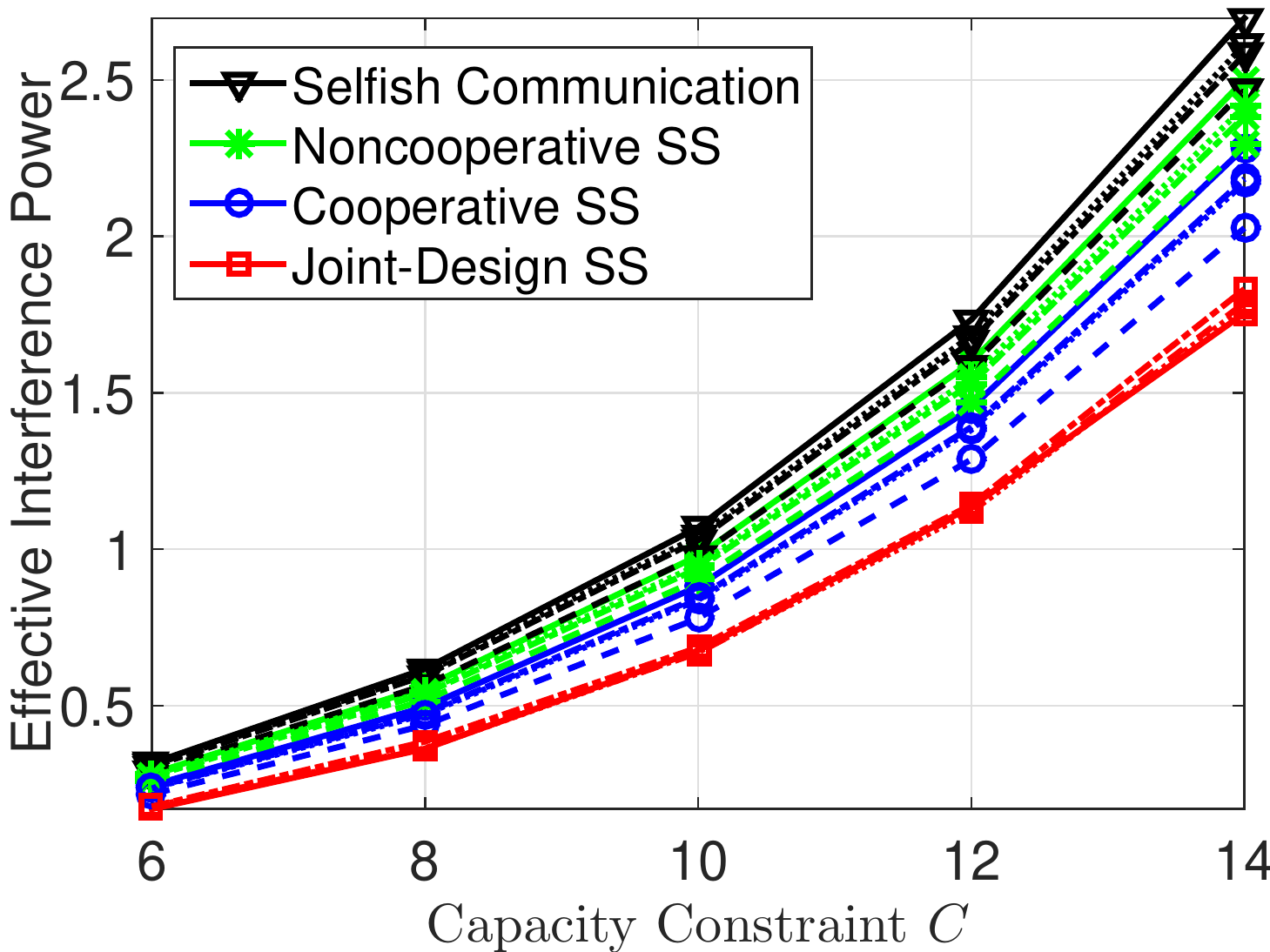}
  }
  \hspace{-5mm}
  \subfigure{
  \includegraphics[width=4.3cm]{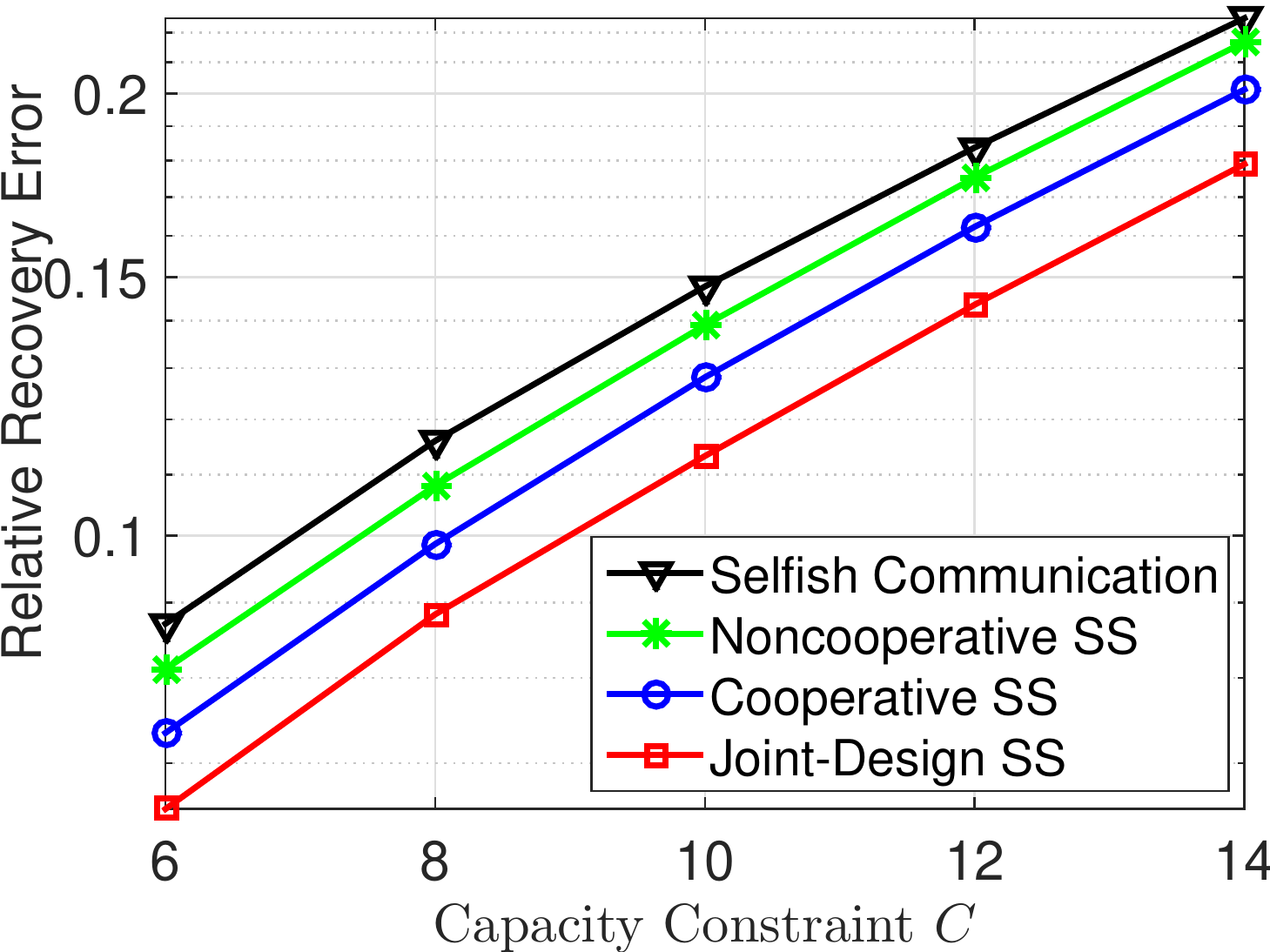}
  }
  \vspace{-5mm}
  \caption{Spectrum sharing with the Scheme I radar under different capacity constraints $C$. $M_{t,R}=16,M_{r,R}=32, M_{t,C}=M_{r,C}=4$. } \label{fig:SchemeIcapacity2}
\end{figure}
\subsubsection{Performance under different capacity constraints}
In this simulation, the constant $C$ in the communication capacity constraint of (\ref{eqn:constrSR}) varies from $6$ to $14$ bits/symbol, while the sub-sampling rate $p$ is fixed to $0.5$. Four different realizations of $\Omega^0$ are considered. Fig. \ref{fig:SchemeIcapacity1} shows the results for $M_{t,R}=4, M_{r,R}=M_{t,C}=8,M_{r,C}=4$. For the ``selfish communication" and noncooperative SS methods, the EIP and relative recovery errors increase as the communication capacity increases. In contrast, the cooperative and joint-design SS methods achieve significantly smaller EIP and relative recovery errors under all values of $C$. This indicates that the latter two SS methods successfully allocate the communication transmit power in directions that result in high communication rate, but small EIP to the Scheme I radar.

The results for $M_{t,R}=16,M_{r,R}=32, M_{t,C}=M_{r,C}=4$ are shown in Fig. \ref{fig:SchemeIcapacity2}. Since $M_{r,R}$ is much larger than $M_{t,C}$, the cooperative SS method outperforms the noncooperative counterpart only marginally. Meanwhile, the joint-design SS method can effectively further reduce the EIP and relative recovery errors.

\begin{figure}
  \centering
  \includegraphics[width=6cm]{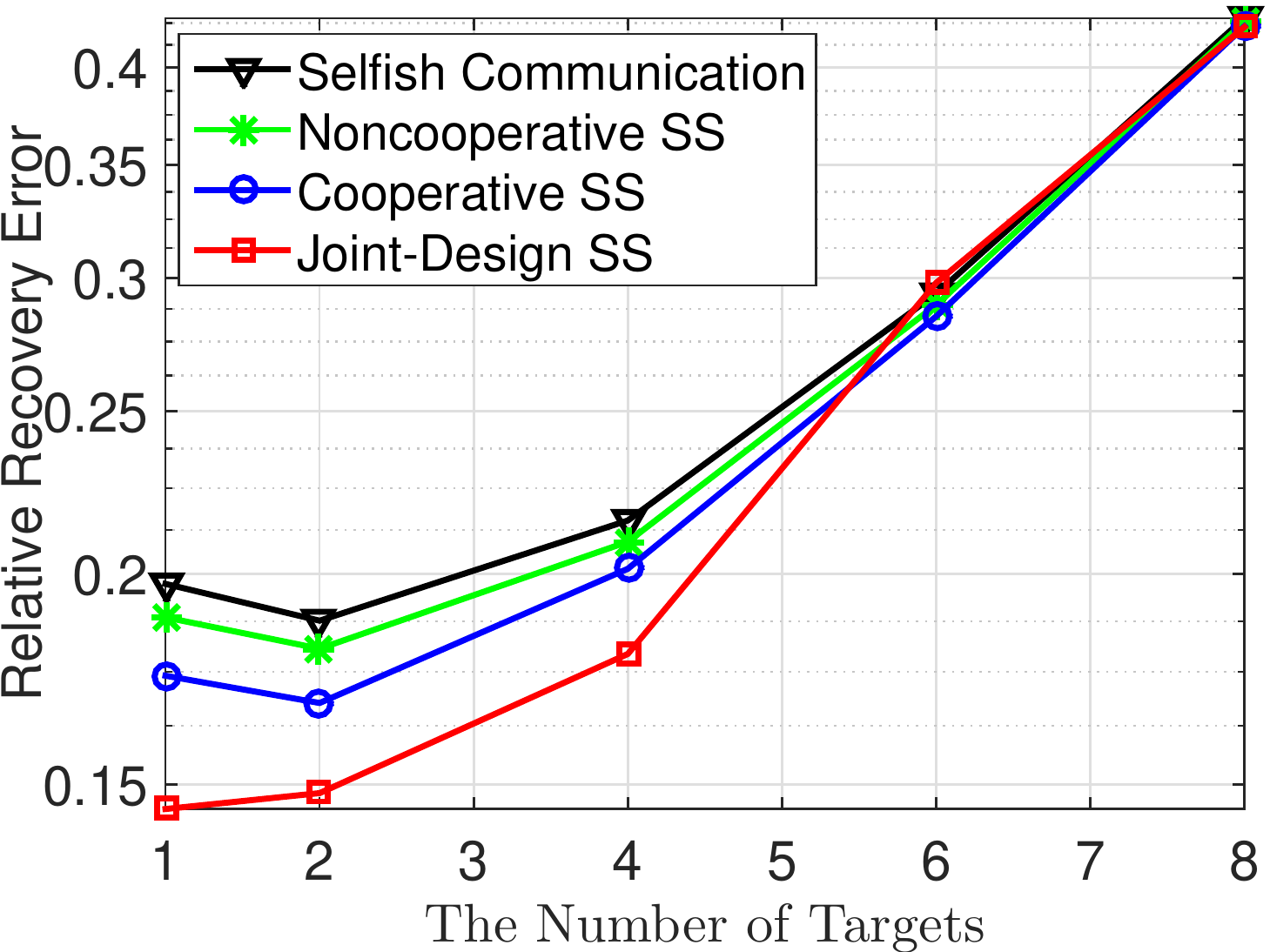}\\
  \caption{Spectrum sharing with the Scheme I radar when multiple targets present. $M_{t,R}=16,M_{r,R}=32, M_{t,C}=M_{r,C}=4$, $p=0.5$ and $C=12$ bits/symbol.}\label{fig:SchemeIntarget}
\end{figure}

\subsubsection{Performance under different number of targets}
In this simulation, we fix $p=0.5$ and $C=12$ and evaluate the performance when multiple targets are present. The target reflection coefficients are designed such that the target returns have fixed power, independent of the number of targets. We observe that the EIPs of different methods remain constant for different number of targets. This is because the design of the communication waveforms is not affected by the target number. Fig. \ref{fig:SchemeIntarget} shows the results of the relative recovery error, which increases as the number of targets increases. All methods have large recovery error for large number of targets, because the retained samples are not sufficient for reliable matrix completion under any level of noise. The proposed joint-design SS method can work effectively for the Scheme I radar when a moderate number of targets are present.

\begin{figure}[htb]
\vspace{-2mm}
  \centering
  \subfigure{
  \includegraphics[width=4.3cm]{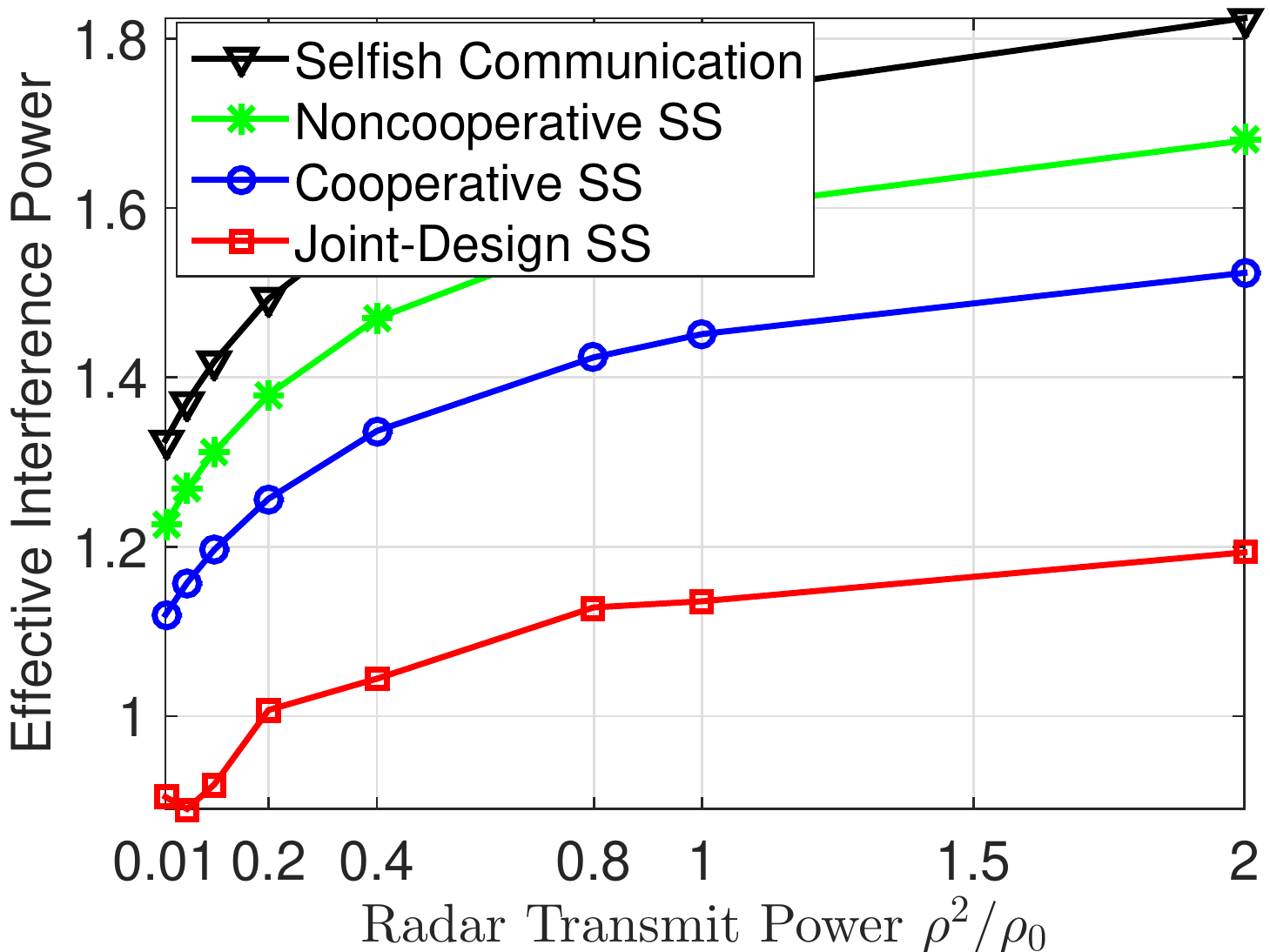}
  }
  \hspace{-5mm}
  \subfigure{
  \includegraphics[width=4.3cm]{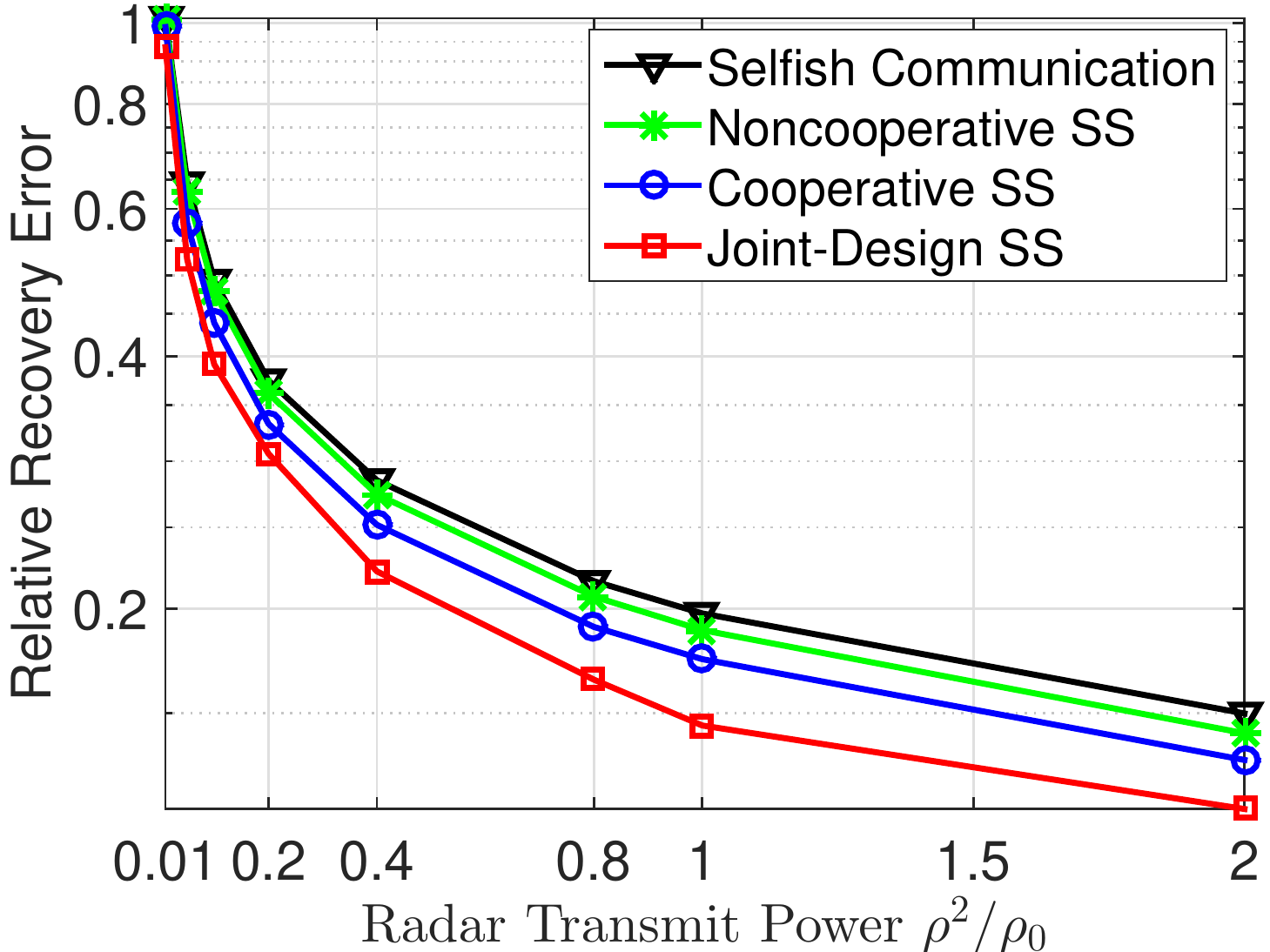}
  }
  \vspace{-5mm}
  \caption{Spectrum sharing with the Scheme I radar under different levels of radar TX power. $M_{t,R}=16, M_{r,R}=32,M_{t,C}=M_{r,C}=4$.} \label{fig:SchemeITXpower}
\end{figure}
\subsubsection{Performance under different levels of radar TX power}
In this simulation, we evaluate the effect of radar TX power $\rho_2$, while fixing $p=0.5$, $C=12$ and the target number to be $1$. Fig. \ref{fig:SchemeITXpower} shows the results of EIP and relative recovery errors for $M_{t,R}=16, M_{r,R}=32,M_{t,C}=M_{r,C}=4$. Again, we see that the joint-design SS method performs the best, followed by the cooperative and then the noncooperative one. When the radar TX power increases, the EIP increases but with a much slower rate. Therefore, increasing radar TX power improves the relative recovery errors.

\begin{figure}[htb]
\vspace{-2mm}
  \centering
  \subfigure{
  \includegraphics[width=4.3cm]{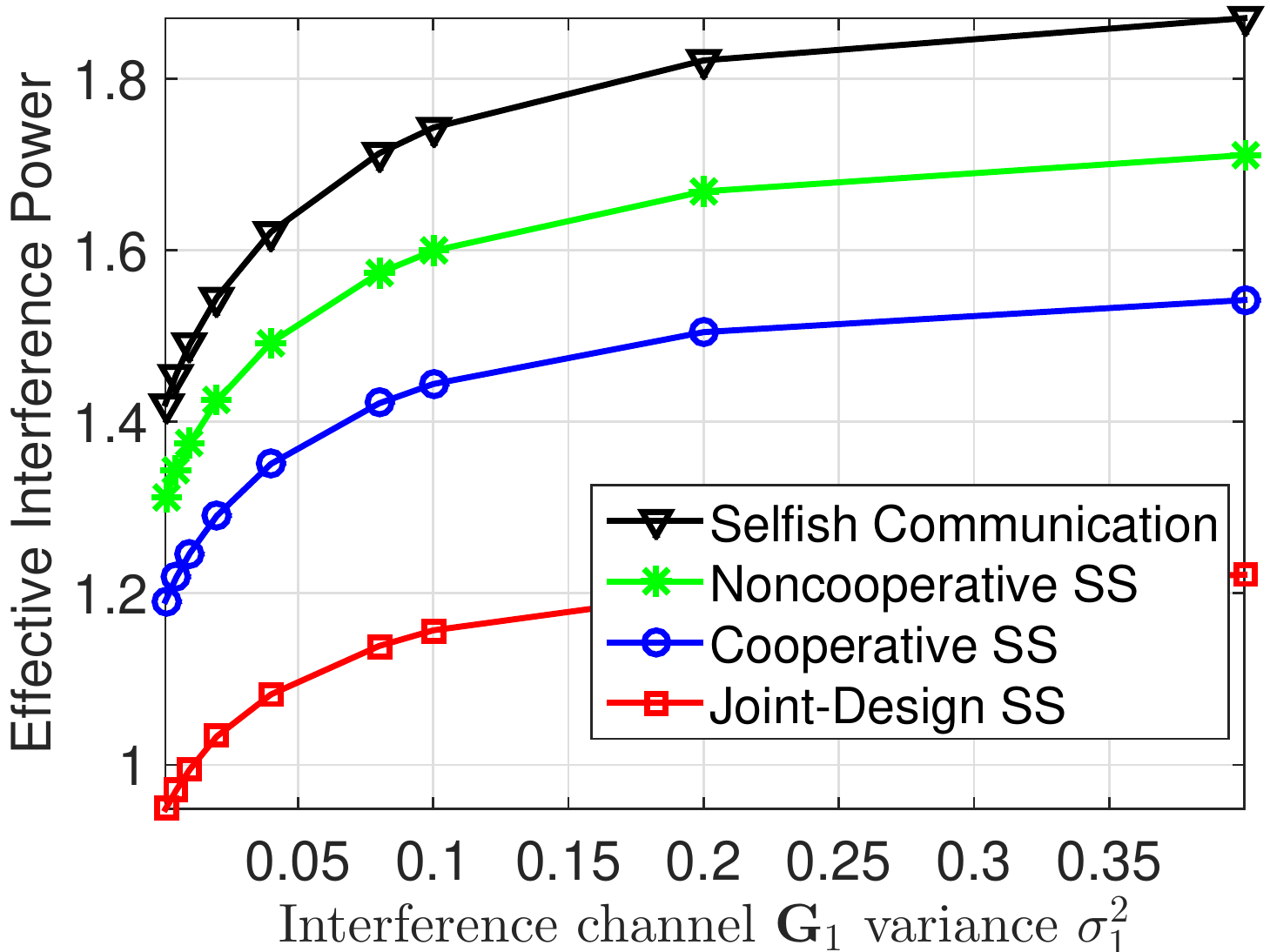}
  }
  \hspace{-5mm}
  \subfigure{
  \includegraphics[width=4.3cm]{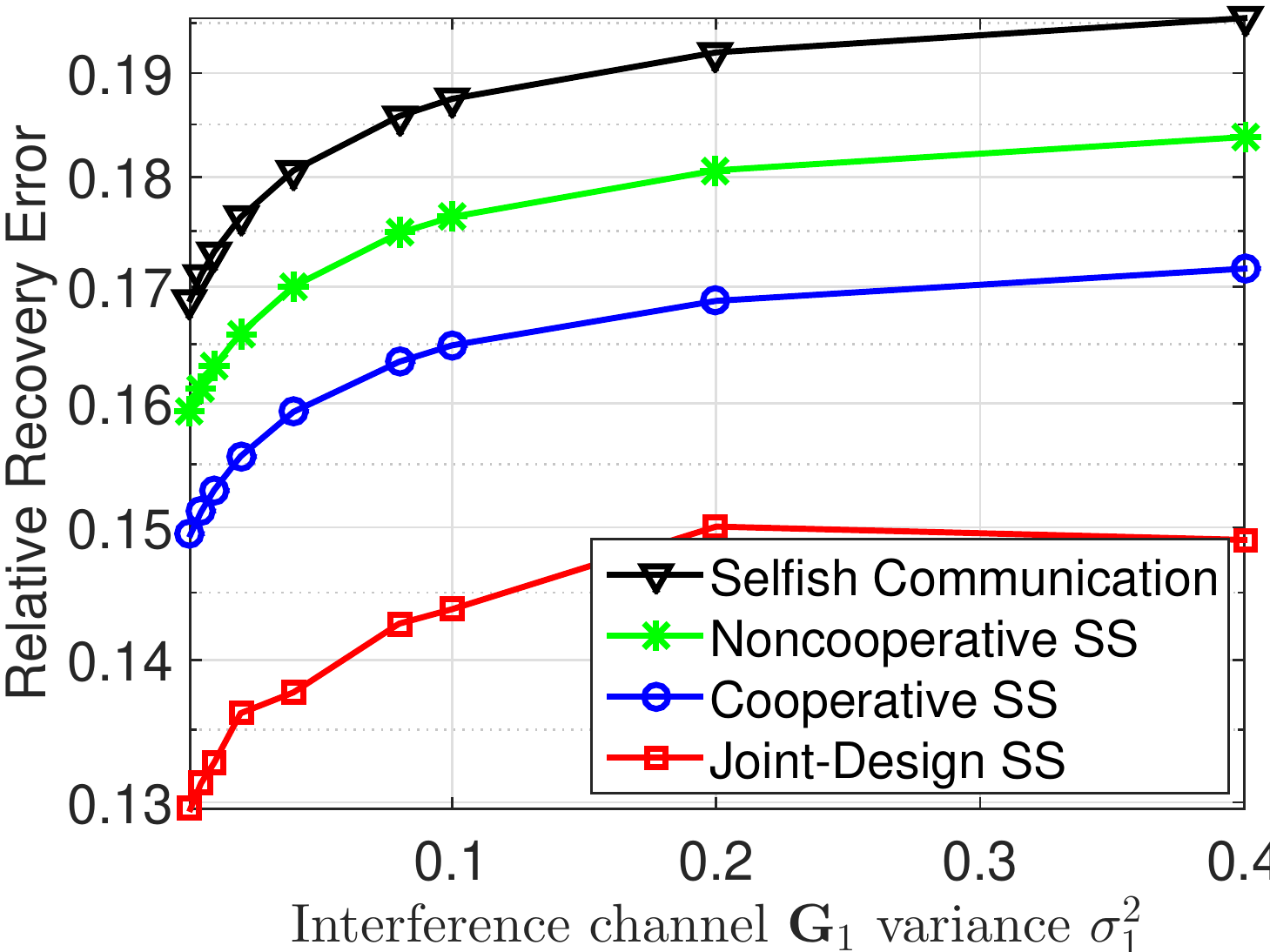}
  }
  \vspace{-5mm}
  \caption{Spectrum sharing with the Scheme I radar under different channel variance $\sigma_1^2$ for the interference channel $\mathbf{G}_1$. $M_{t,R}=16, M_{r,R}=32,M_{t,C}=M_{r,C}=4$.} \label{fig:SchemeIG1var}
\end{figure}
\subsubsection{Performance under different interference channel strength}
In this simulation, we evaluate the effect the interference channel $\mathbf{G}_1$ with different $\sigma_1^2$, while fixing $p=0.5$, $C=12$ and the target number to be $1$. As the communication RX gets closer to the radar TX antennas, $\sigma_1^2$ gets larger. Fig. \ref{fig:SchemeIG1var} shows the results of EIP and relative recovery errors for $M_{t,R}=16, M_{r,R}=32,M_{t,C}=M_{r,C}=4$. For all the SS methods, when the interference channel $\mathbf{G}_1$ gets stronger, the communication TX increases its transmit power in order to satisfy the capacity constraint. Therefore, the EIP and the relative recovery errors increases with the variance $\sigma_1^2$. We also observe that the joint-design SS method performs the best, followed by the cooperative and then the noncooperative one.

\subsection{Performance of the Scheme II radar and a MIMO Communication System Spectrum Sharing}
\begin{figure}[htb]
\vspace{-2mm}
  \centering
  \subfigure{
  \includegraphics[width=4.3cm]{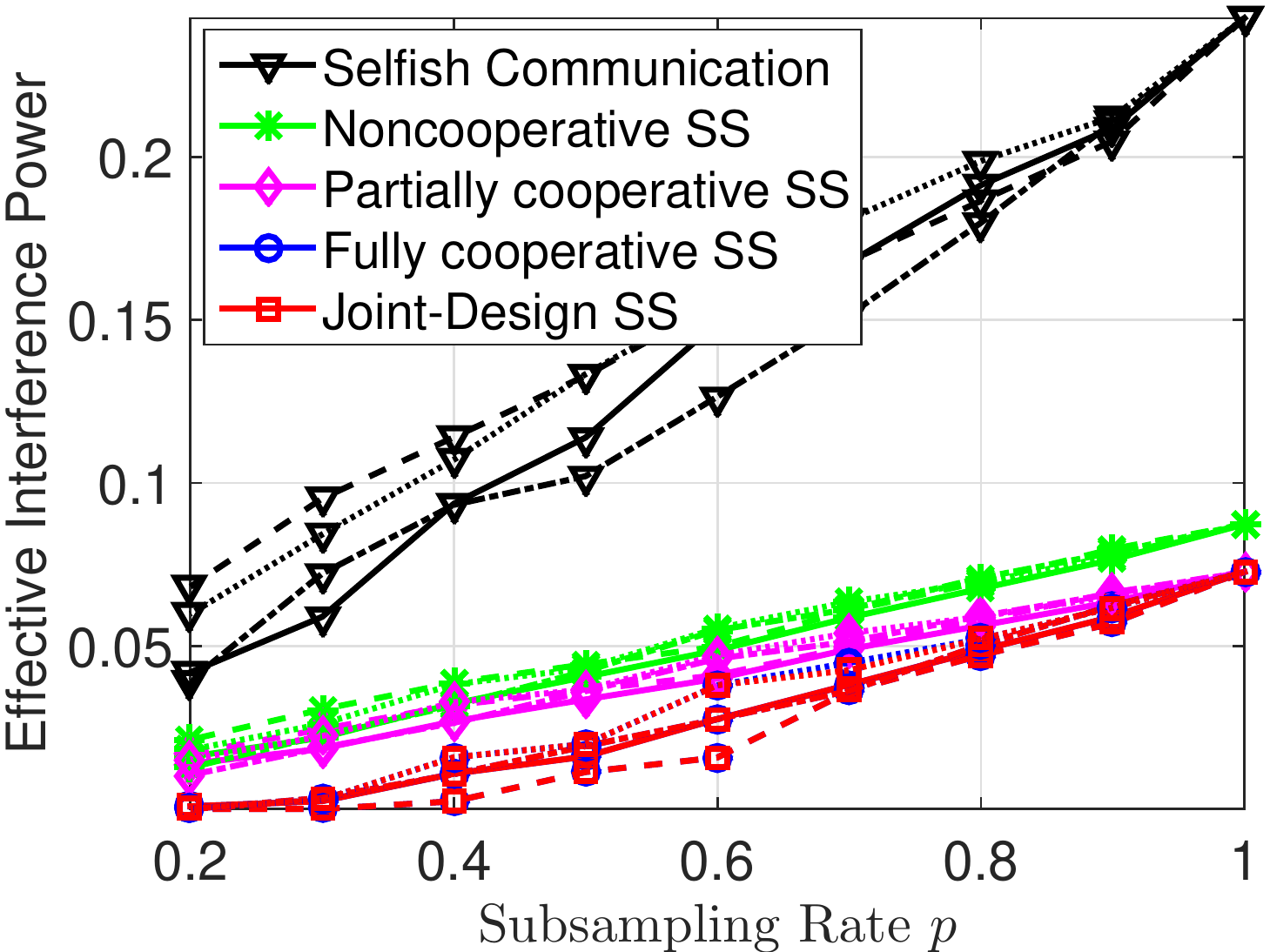}
  }
  \hspace{-5mm}
  \subfigure{
  \includegraphics[width=4.3cm]{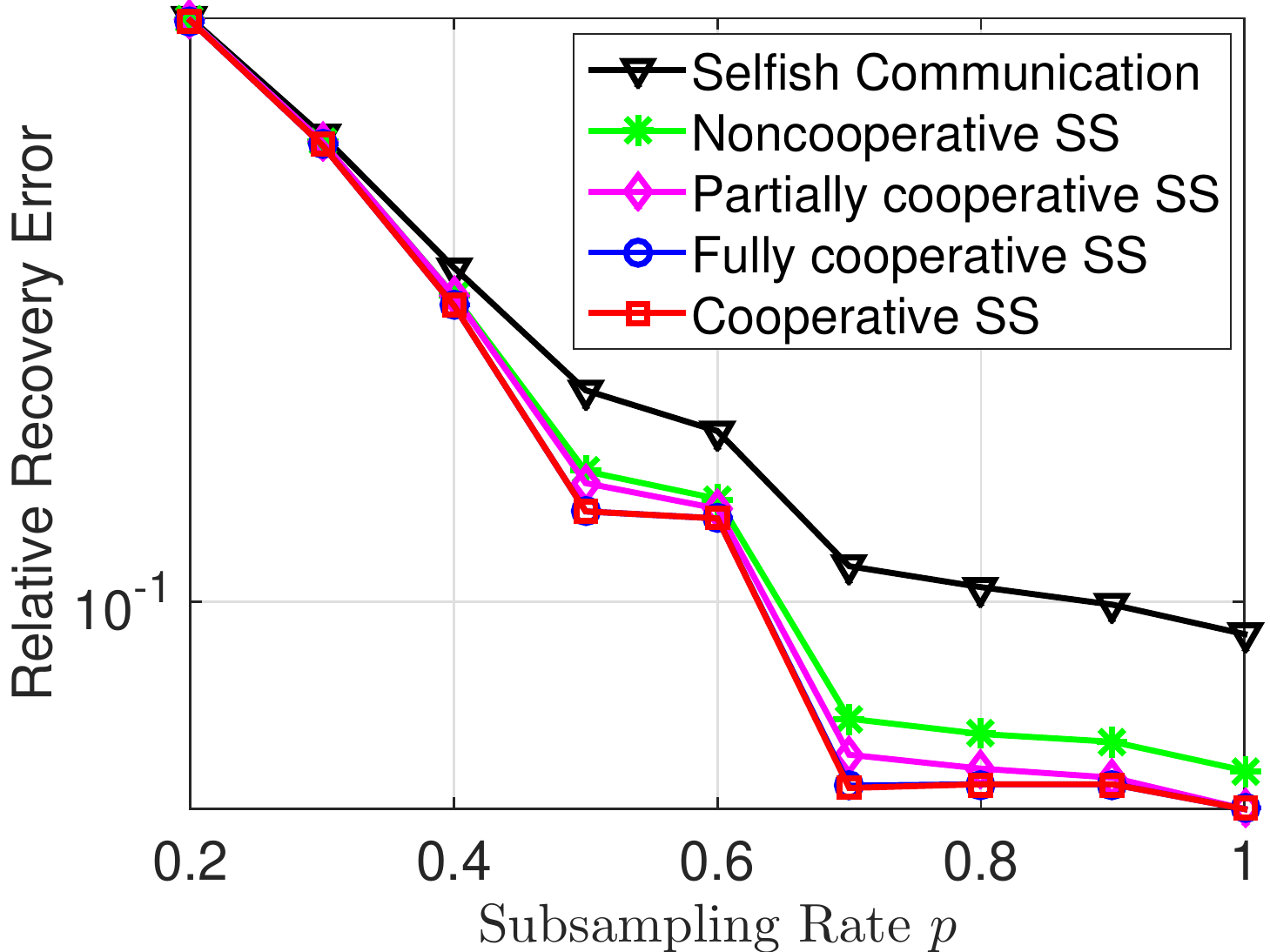}
  }
  \vspace{-5mm}
  \caption{Spectrum sharing with the Scheme II radar under different sub-sampling rates. $M_{t,R}=4, M_{r,R}=M_{t,C}=8,M_{r,C}=4$.} \label{fig:SchemeIIpercent1}
\end{figure}

\begin{figure}[htb]
\vspace{-2mm}
  \centering
  \subfigure{
  \includegraphics[width=4.3cm]{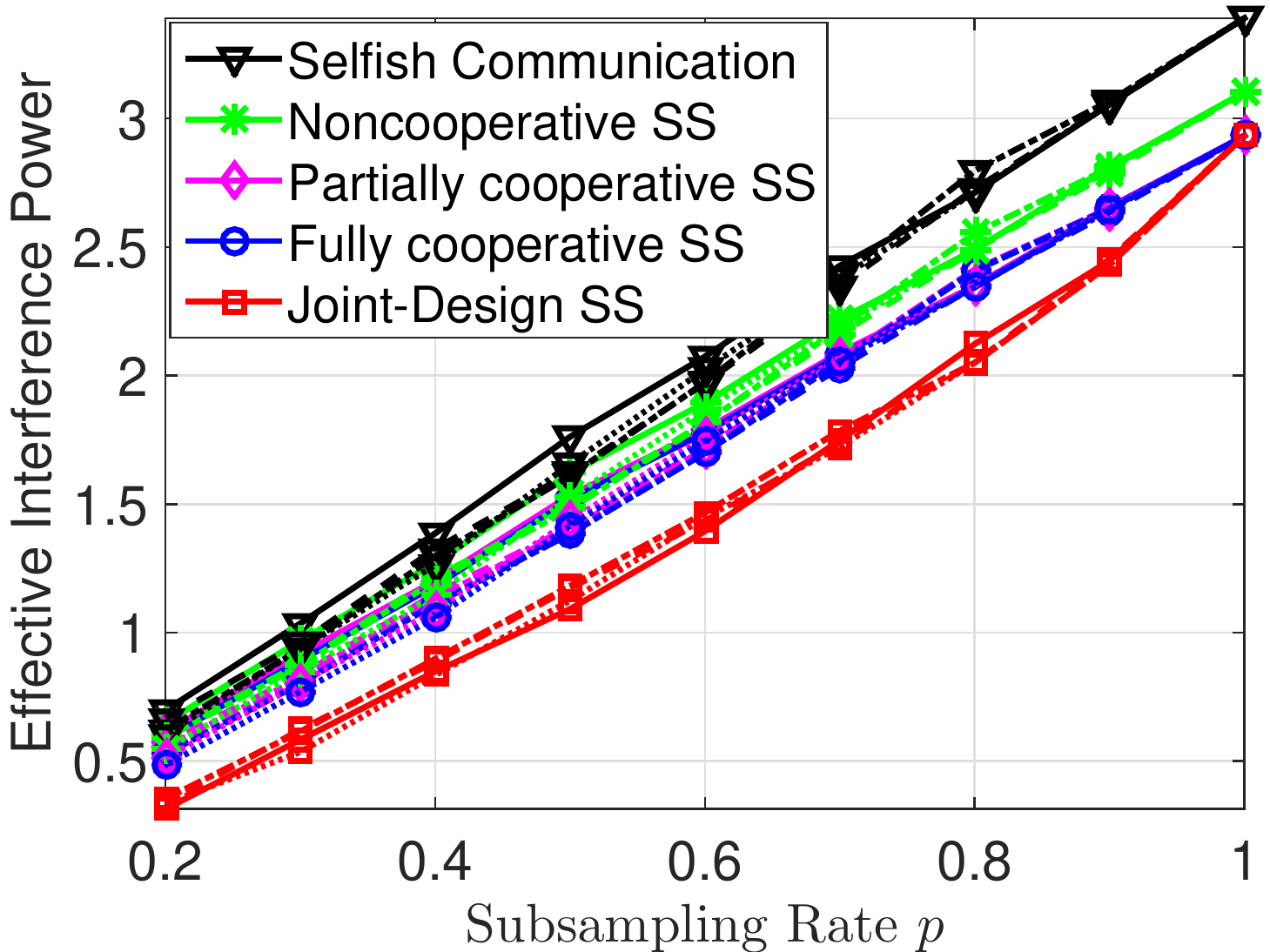}
  }
  \hspace{-5mm}
  \subfigure{
  \includegraphics[width=4.3cm]{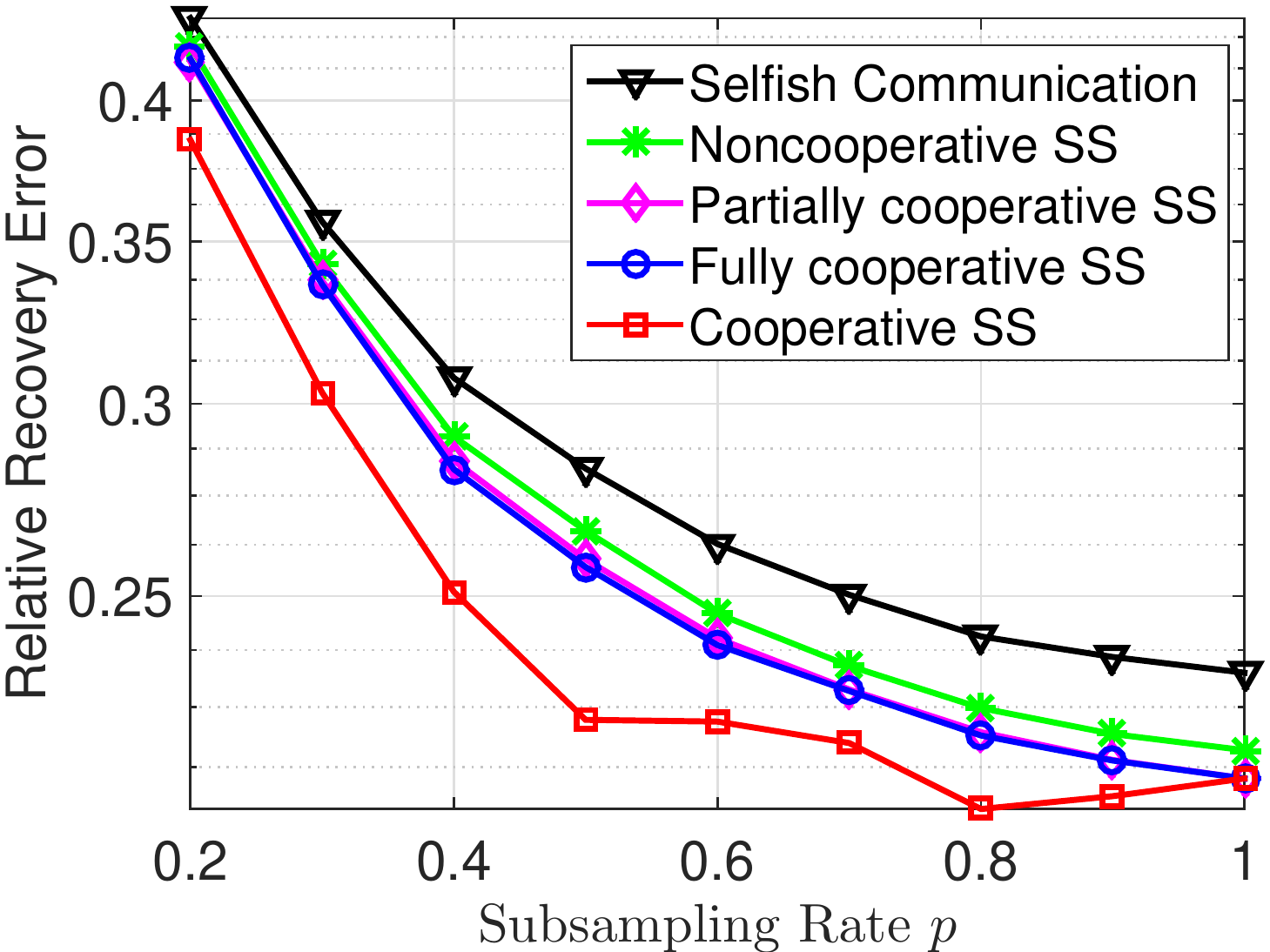}
  }
  \vspace{-5mm}
  \caption{Spectrum sharing with the Scheme II radar under different sub-sampling rates. $M_{t,R}=16,M_{r,R}=32, M_{t,C}=M_{r,C}=4$.} \label{fig:SchemeIIpercent2}
\end{figure}
\subsubsection{Performance under different sub-sampling rates}
Here we consider one far-field stationary target at angle $30^\circ$ w.r.t. the radar arrays, with target reflection coefficient equal to $0.1e^{j\pi/4}$.
For the communication capacity and power constraints, we consider $C = 12$ bits/symbol and $P_t = L$. Again, the sub-sampling rate of Scheme II radar varies from $0.2$ to $1$.
The following two scenarios are considered.

In the first scenario, we consider $M_{t,R}=4, M_{r,R}=M_{t,C}=8,M_{r,C}=4$. The EIP to the Scheme II radar and the relative recovery errors of the matrix completion are shown in Fig. \ref{fig:SchemeIIpercent1}. The EIPs are shown for four realizations of $\mathbf{\Omega}^0$, while the relative recovery errors are the average over the realizations of $\mathbf{\Omega}^0$. We observe that all four proposed SS methods achieve significant interference reduction compared to the ``selfish communication method". Higher level cooperation between the MIMO-MC radar and communication systems achieves greater EIP reduction. The fully cooperative and joint-design SS methods outperform their partially cooperative, and noncooperative counterparts, which validates the statement in Theorem \ref{thm:MIMO-MCII}. However, for small $p$'s, the improvement achieved by the fully cooperative SS method is not as significant as that when the Scheme I radar is considered (see Fig. \ref{fig:SchemeIpercent1}). This is reasonable because $\mathbf{\Delta}_{l\xi}$ in the expression of $\text{EIP}_{II}$ is always full rank even for small values of $p$. Decreasing $p$ will not reduce the rank of effective reference channel $\sqrt{\mathbf{\Delta}_{l\xi}}\mathbf{G}_2$. Therefore, the communication system cannot find a direction that would introduce zero EIP to the radar.

In the second scenario, we consider $M_{t,R}=16, M_{r,R}=32, M_{t,C}=M_{r,C}=4$. Fig. \ref{fig:SchemeIIpercent2} shows the EIP and the relative recovery errors of the matrix completion. Again, the EIPs are shown for four realizations of $\mathbf{\Omega}^0$, while the relative recovery errors are the average over the realizations of $\mathbf{\Omega}^0$. The joint-design SS method achieves much smaller EIP and relative recovery errors than the other four methods. This validates the effectiveness of the proposed joint-design SS method for the Scheme II radar. We conclude that the MC approach benefits the Scheme II radar by  reducing not only the data to be forwarded to the fusion center but also the effective interference from the communication system when spectrum sharing is considered.

\begin{figure}[htb]
\vspace{-2mm}
  \centering
  \subfigure{
  \includegraphics[width=4.3cm]{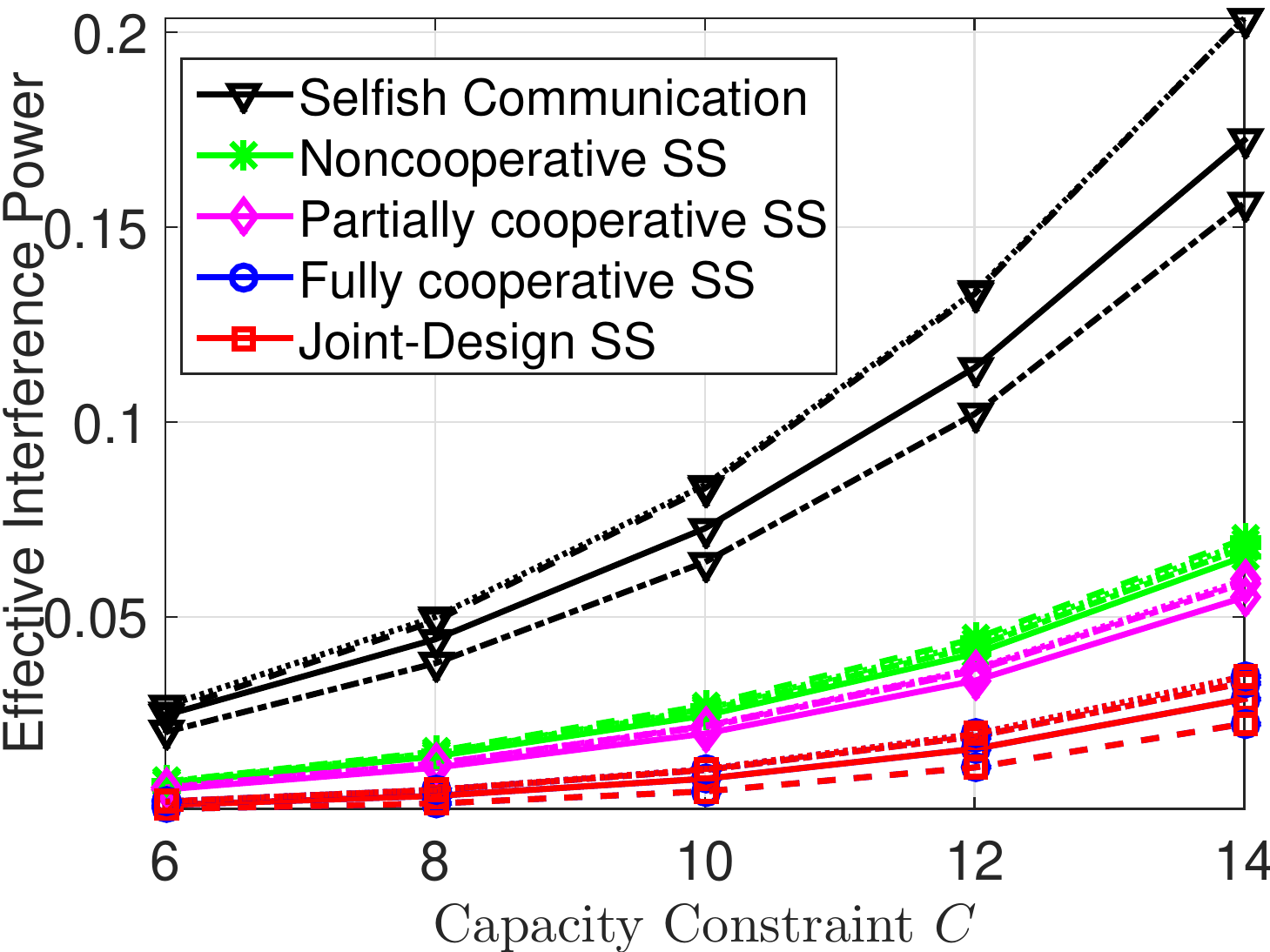}
  }
  \hspace{-5mm}
  \subfigure{
  \includegraphics[width=4.3cm]{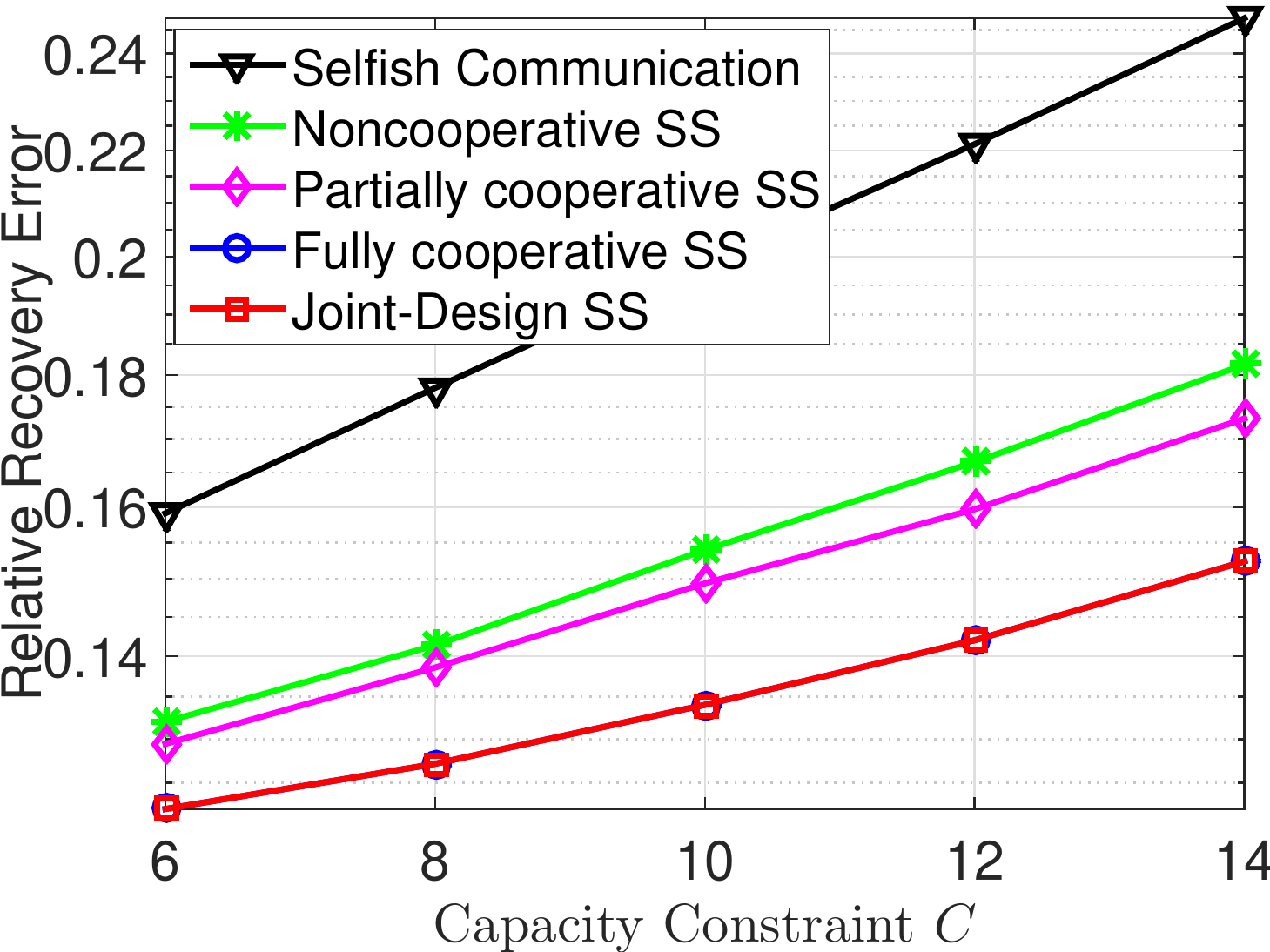}
  }
  \vspace{-5mm}
  \caption{Spectrum sharing with the Scheme II radar under different capacity constraints $C$. $M_{t,R}=4, M_{r,R}=M_{t,C}=8,M_{r,C}=4$.} \label{fig:SchemeIIcapacity1}
\end{figure}

\begin{figure}[htb]
\vspace{-2mm}
  \centering
  \subfigure{
  \includegraphics[width=4.3cm]{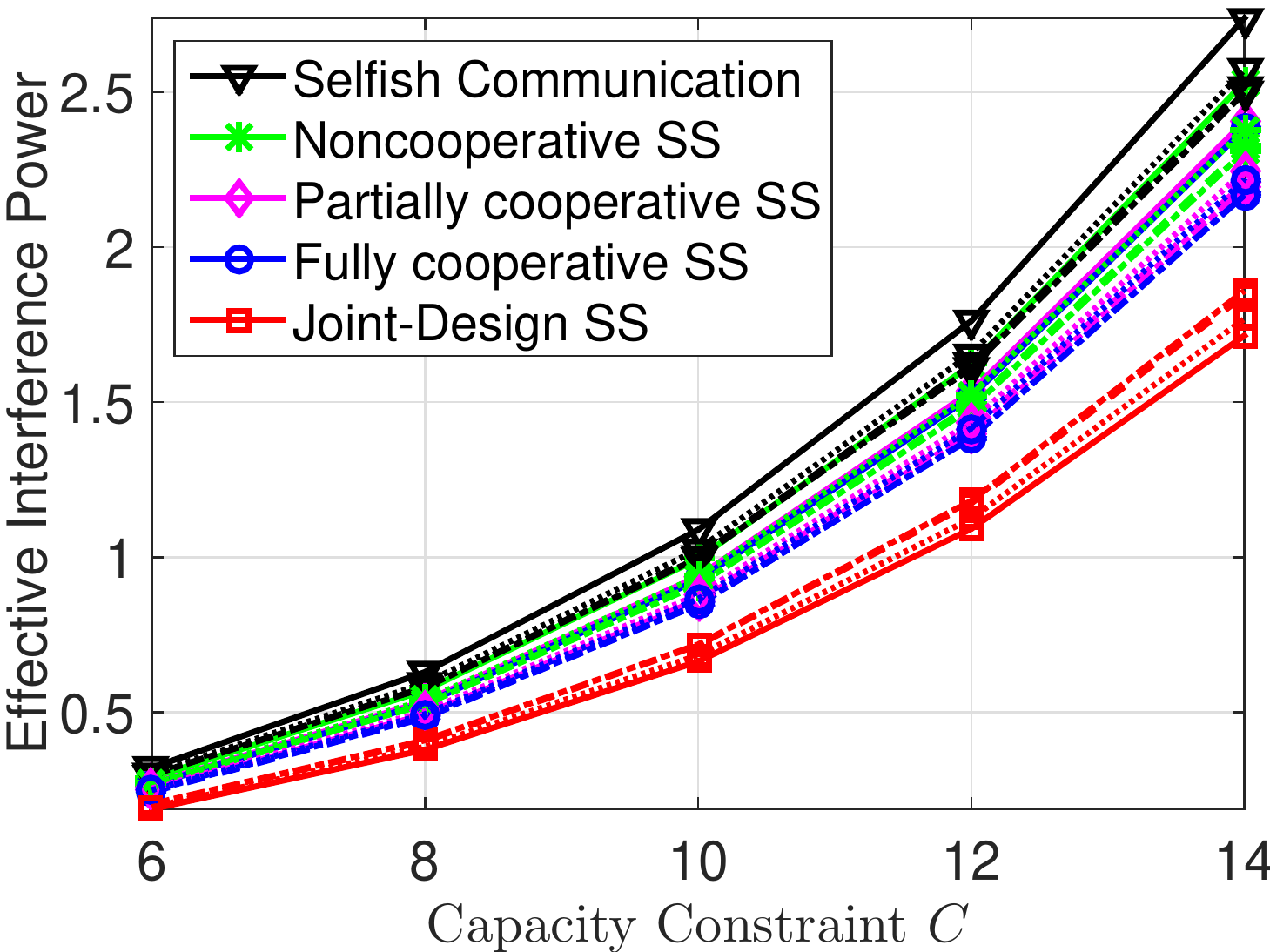}
  }
  \hspace{-5mm}
  \subfigure{
  \includegraphics[width=4.3cm]{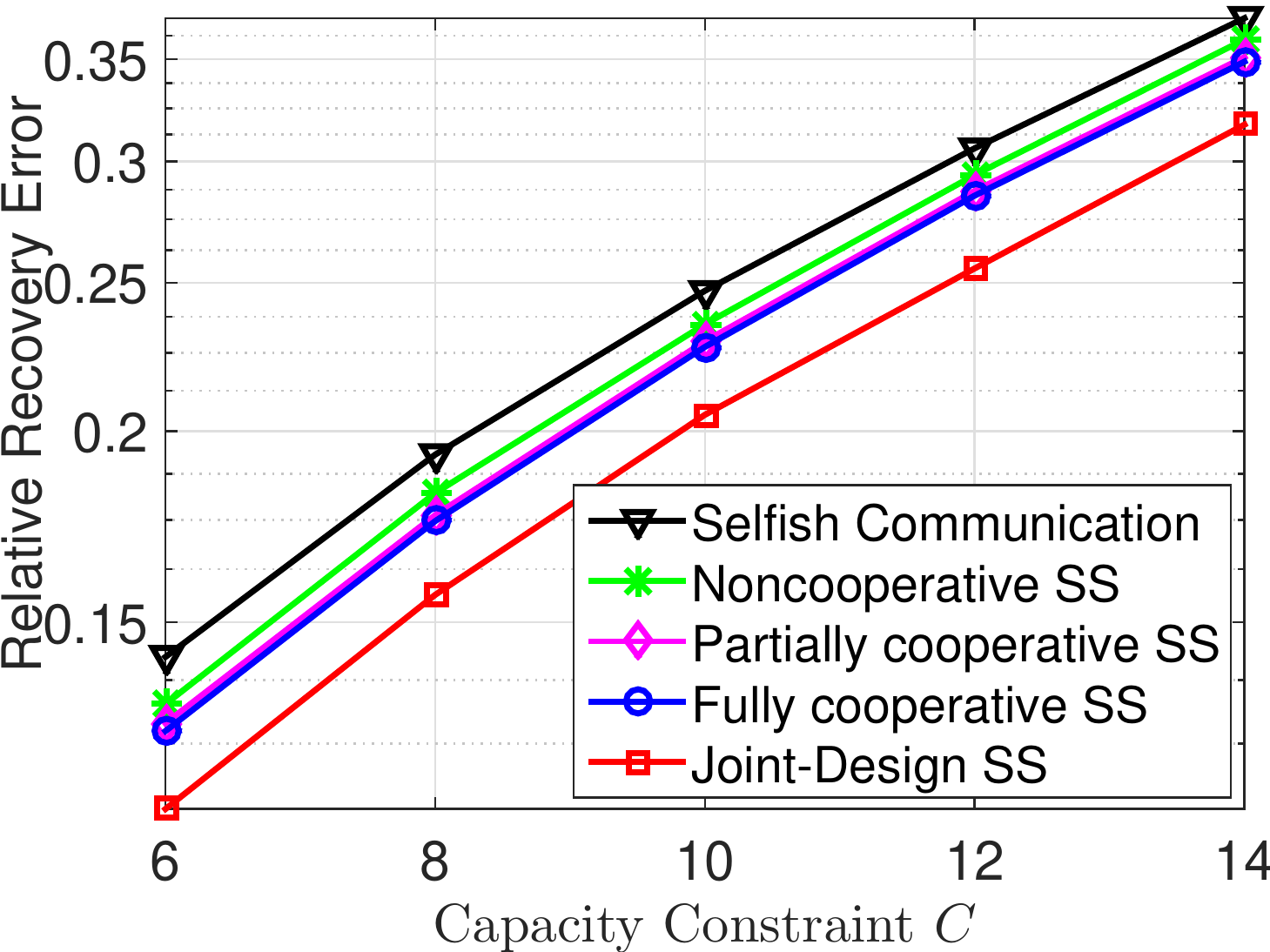}
  }
  \vspace{-5mm}
  \caption{Spectrum sharing with the Scheme II radar under different capacity constraints $C$. $M_{t,R}=16,M_{r,R}=32, M_{t,C}=M_{r,C}=4$. } \label{fig:SchemeIIcapacity2}
\end{figure}
\subsubsection{Performance under different capacity constraints}
In this simulation, the Scheme II radar has fixed sub-sampling rate $p=0.5$, while the communication capacity constant $C$ in (\ref{eqn:constrSR}) varies from $6$ to $14$ bits/symbol. Similarly, two scenarios are considered. The results for the scenario of $M_{t,R}=4, M_{r,R}=M_{t,C}=8,M_{r,C}=4$ are plotted in Fig. \ref{fig:SchemeIIcapacity1}, and the scenario of $M_{t,R}=16,M_{r,R}=32, M_{t,C}=M_{r,C}=4$ in Fig. \ref{fig:SchemeIIcapacity2}. The ``selfish communication" method is inferior to all proposed SS methods. One can also observe that the joint-design SS method always achieves considerably smaller EIP and relative recovery errors than the partially and noncooperative SS methods.

\begin{figure}
  \centering
  \includegraphics[width=6cm]{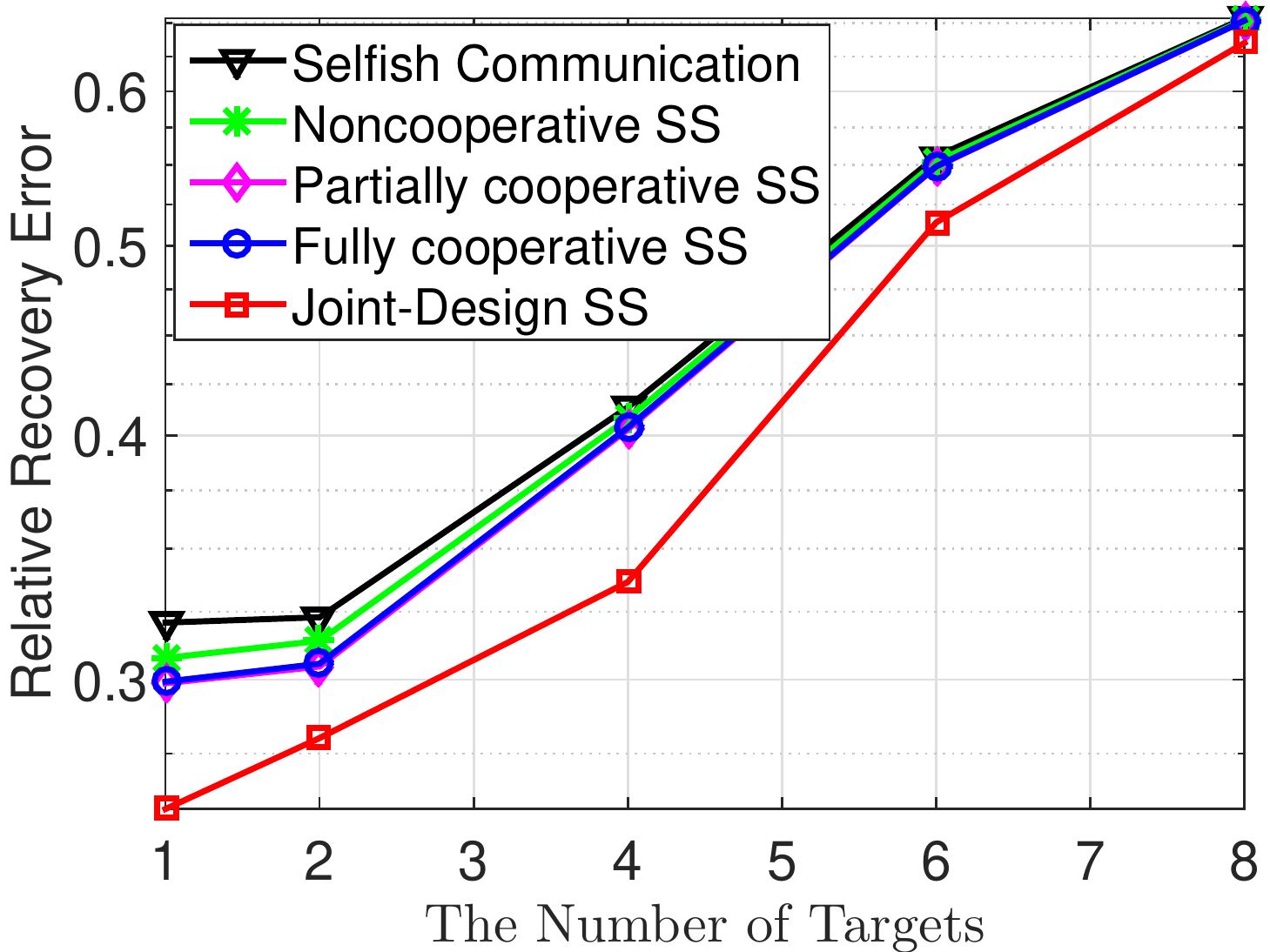}\\
  \caption{Spectrum sharing with the Scheme II radar when multiple targets present. $M_{t,R}=16,M_{r,R}=32, M_{t,C}=M_{r,C}=4$, $p=0.5$ and $C=12$ bits/symbol.}\label{fig:SchemeIIntarget}
\end{figure}

\subsubsection{Performance under different number of targets}
In this simulation, we fix $p=0.5$ and $C=12$ and evaluate the performance when multiple targets are present. The target reflection coefficients are designed such that the target returns have fixed power, independent the number of targets. Again, we observe that the EIPs of different methods remain constant for different number of targets. The results of the relative recovery error are shown in Fig. \ref{fig:SchemeIIntarget}.
The proposed joint-design SS method can work effectively for the Scheme II radar when a moderate number of targets are present and sufficient samples are used for matrix completion.

\begin{figure}[htb]
\vspace{-2mm}
  \centering
  \subfigure{
  \includegraphics[width=4.3cm]{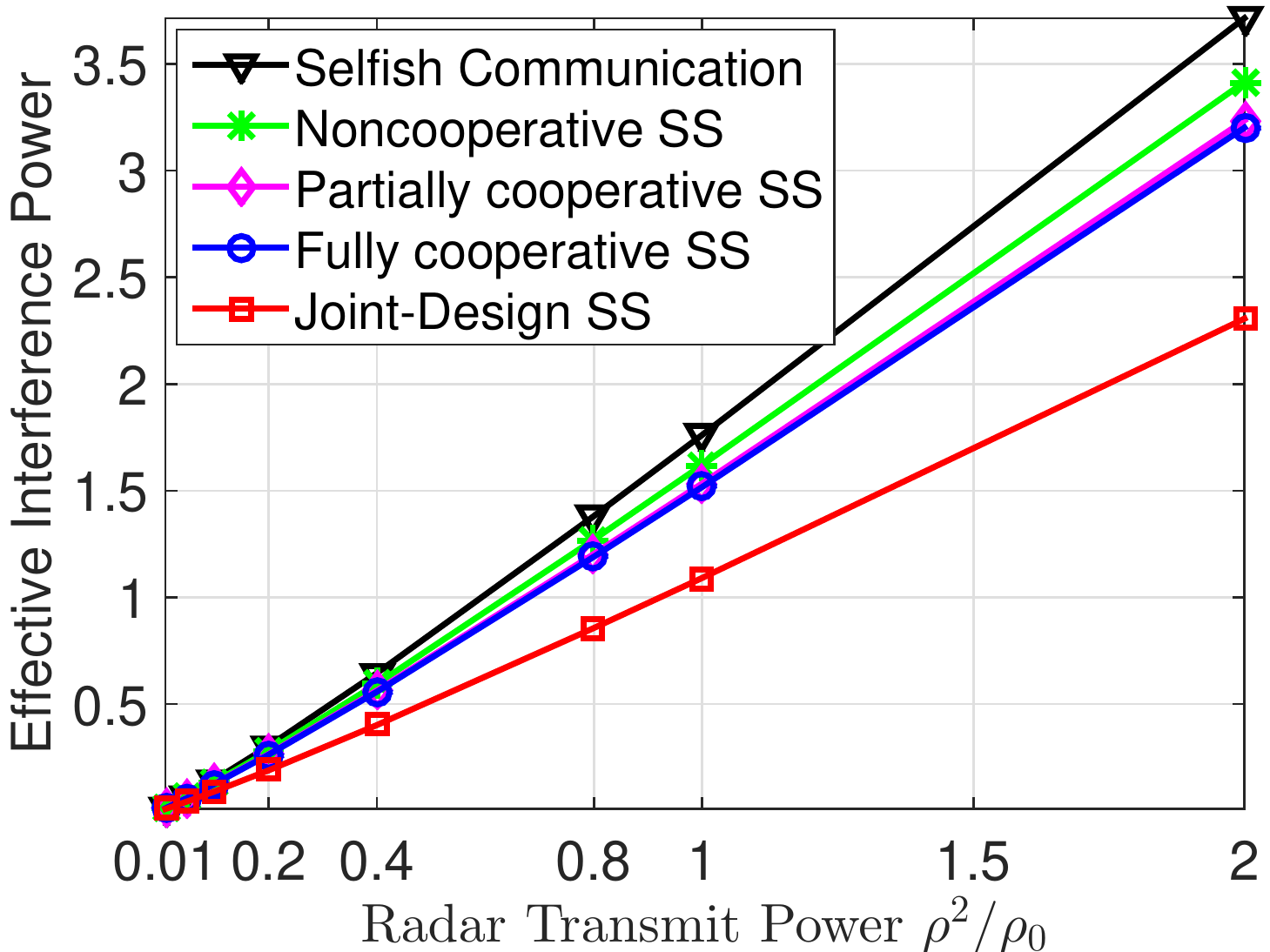}
  }
  \hspace{-5mm}
  \subfigure{
  \includegraphics[width=4.3cm]{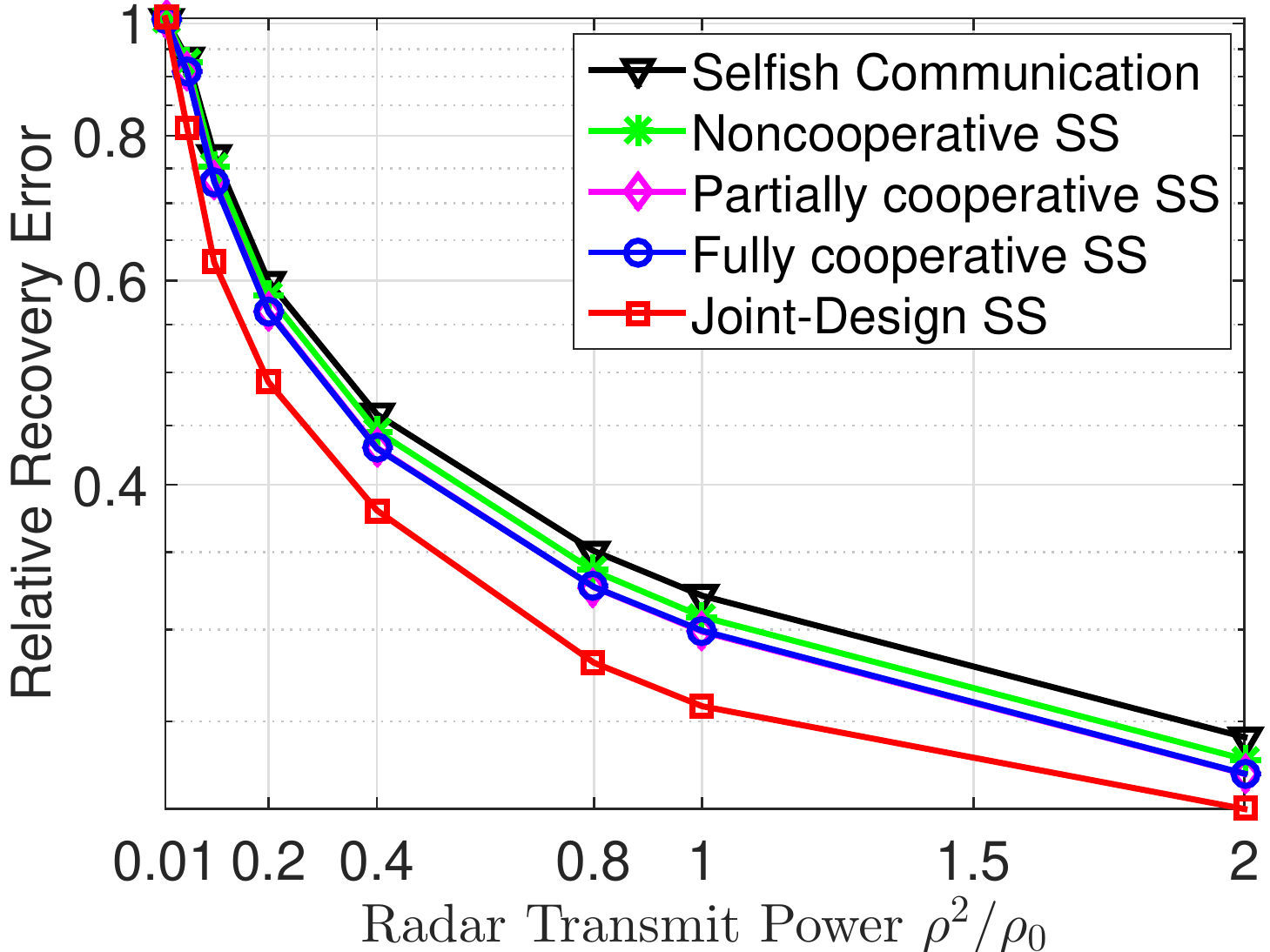}
  }
  \vspace{-5mm}
  \caption{Spectrum sharing with the Scheme II radar under different levels of radar TX power. $M_{t,R}=16, M_{r,R}=32,M_{t,C}=M_{r,C}=4$.} \label{fig:SchemeIITXpower}
\end{figure}
\subsubsection{Performance under different levels of radar TX power}
In this simulation, we evaluate the effect of radar TX power $\rho_2$, while fixing $p=0.5$, $C=12$ and the target number to be $1$. Fig. \ref{fig:SchemeIITXpower} shows the results of EIP and relative recovery errors for $M_{t,R}=16, M_{r,R}=32,M_{t,C}=M_{r,C}=4$. We can see that the joint-design SS method greatly outperforms the other three methods. When the radar TX power increases, the performance gap between the joint-design SS method and the other three methods becomes larger.

\begin{figure}[htb]
\vspace{-2mm}
  \centering
  \subfigure{
  \includegraphics[width=4.3cm]{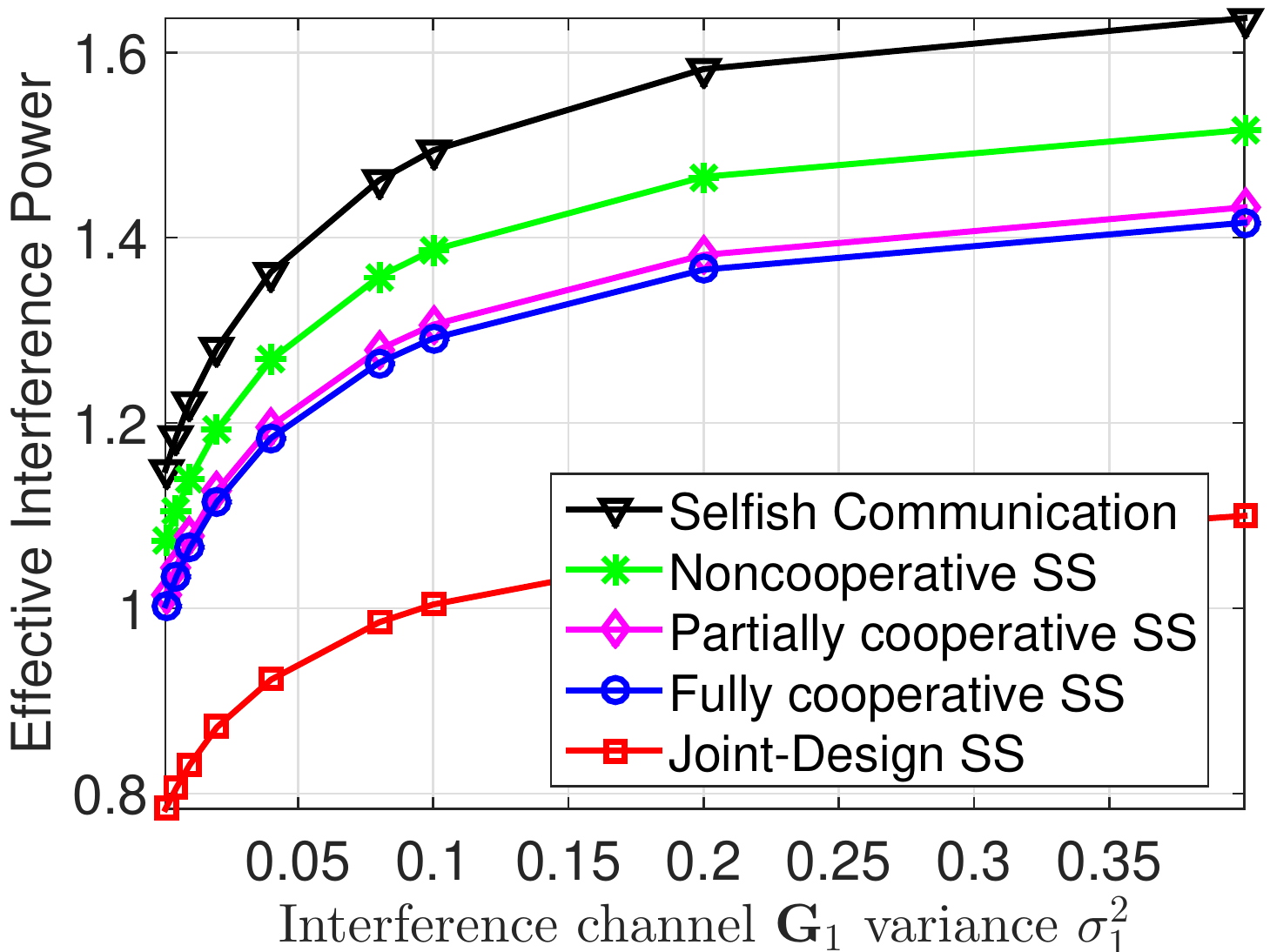}
  }
  \hspace{-5mm}
  \subfigure{
  \includegraphics[width=4.3cm]{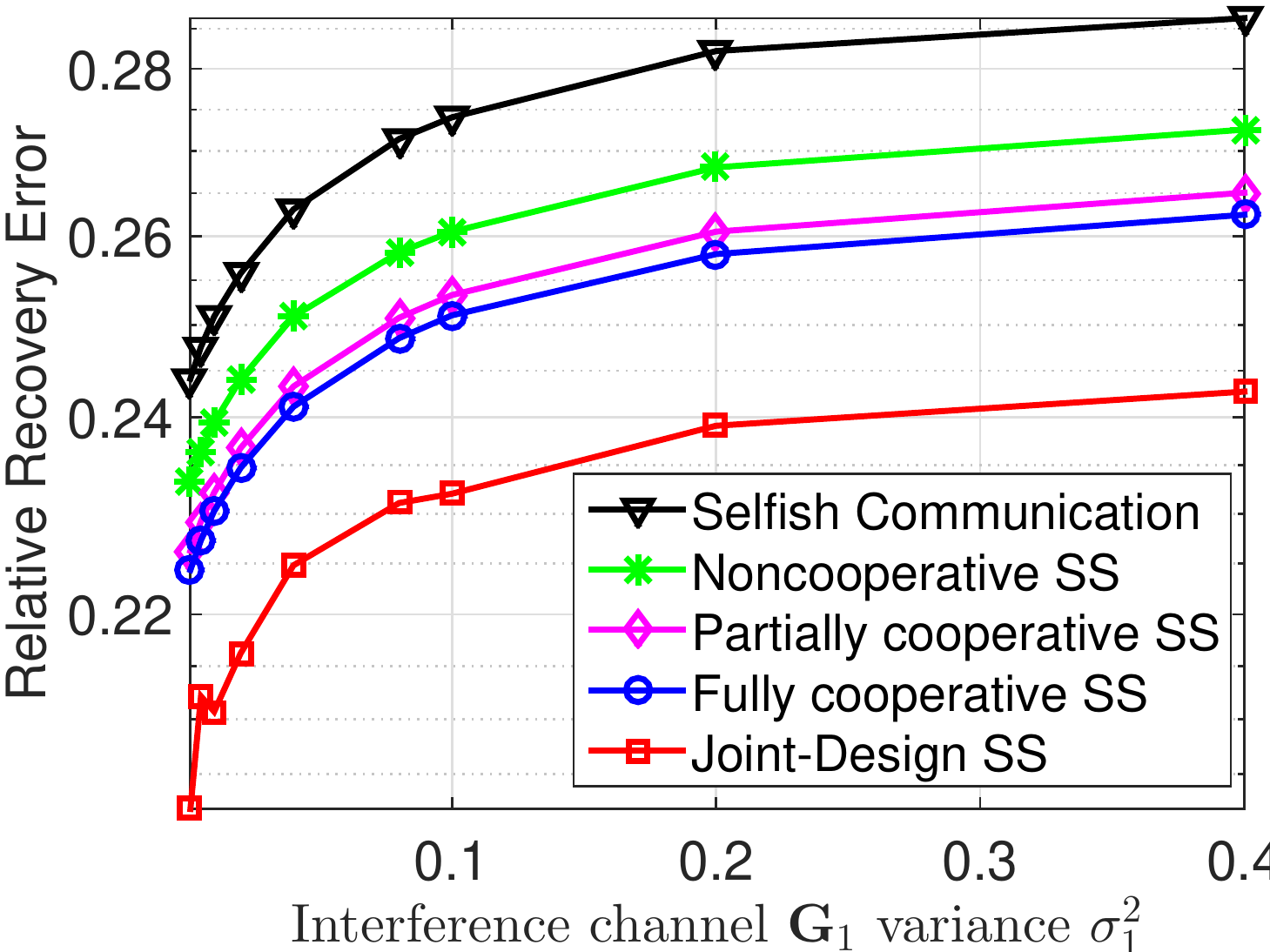}
  }
  \vspace{-5mm}
  \caption{Spectrum sharing with the Scheme II radar under different channel variance $\sigma_1^2$ for the interference channel $\mathbf{G}_1$. $M_{t,R}=16, M_{r,R}=32,M_{t,C}=M_{r,C}=4$.} \label{fig:SchemeIIG1var}
\end{figure}
\subsubsection{Performance under different interference channel strength}
In this simulation, we evaluate the effect the interference channel $\mathbf{G}_1$ with different $\sigma_1^2$, while fixing $p=0.5$, $C=12$ and the target number to be $1$. Fig. \ref{fig:SchemeIIG1var} shows the results of EIP and relative recovery errors for $M_{t,R}=16, M_{r,R}=32,M_{t,C}=M_{r,C}=4$. For all the SS methods, when the interference channel $\mathbf{G}_1$ gets stronger, the communication TX increases its transmit power in order to satisfy the capacity constraint. Therefore, the EIP and the relative recovery errors increases with the variance $\sigma_1^2$. We also observe that the joint-design SS method greatly outperforms the other three methods.

\section{Conclusions} \label{sec:conclusion}
This paper has considered  spectrum sharing (SS) between a MIMO communication system and a MIMO-MC radar system using two different schemes, Scheme I and Scheme II. In order to reduce the effective interference power (EIP) at radar RX antennas, we have first proposed two communication transmit covariance matrix design strategies, namely, a noncooperative and a cooperative SS method, for both Schemes I and II.
Our theoretical results guarantee that the cooperative approach can effectively reduce the EIP to a larger extent as compared to the noncooperative approach.
Second, we have proposed a joint design of the communication transmit covariance matrix and the radar sampling scheme to further reduce the EIP. The EIP reduction and the matrix completion recovery errors have been evaluated under various system parameters.
We have shown that both Scheme I and II radars enjoy reduced interference by the communication system when the proposed SS methods are considered.
In particular for Scheme I, the sparse sampling at the radar RX antennas can reduce the rank of the interference channel. Our simulations have confirmed that significant EIP reduction is achieved by the cooperative approach; this is because in that approach, the communication power is allocated to directions in the null space of the effective interference channel.
When the number of radar RX antennas is much larger than that of the communication TX antennas, the cooperative approach outperforms the noncooperative one only marginally. Our simulations have suggested that for both Schemes I and II, the joint-design SS method can achieve much smaller EIP and relative recovery errors than other methods when the number of radar TX and RX antennas is moderately large.

\appendices
\section{Proof of Lemma \ref{lemma:EIPII1}} \label{appen:EIPII1}
Let us first look at the interference power at the $m$-th radar receive antenna, $m\in\mathbb{N}_{M_{r,R}}^+$,
\begin{equation}
\begin{aligned}
\beta_m&\triangleq \mathbb{E}\left\{ \mathbf{g}_m^H\mathbf{X}\mathbf{\Lambda}_2\mathbf{S}_m^H \mathbf{S}_m\mathbf{\Lambda}_2^H\mathbf{X}^H \mathbf{g}_m \right\} \\
&= \mathbb{E}\left\{\text{Tr} \left(\mathbf{g}_m^H\mathbf{X}\mathbf{\Lambda}_2\mathbf{S}_m^H \mathbf{S}_m\mathbf{\Lambda}_2^H\mathbf{X}^H \mathbf{g}_m  \right) \right\} \\
&= \mathbb{E}\Bigl\{\text{Tr} \Bigl(\underbrace{\mathbf{S}_m^H \mathbf{S}_m}_{\triangleq\mathbf{A}_m} \mathbf{\Lambda}_2^H\mathbf{X}^H \underbrace{\mathbf{g}_m \mathbf{g}_m^H}_{\triangleq\mathbf{B}_m}\mathbf{X} \mathbf{\Lambda}_2  \Bigr) \Bigr\} \\
&= \text{Tr} \left(\mathbf{A}_m \mathbb{E}\left\{\mathbf{\Lambda}_2^H\mathbf{X}^H \mathbf{B}_m \mathbf{X} \mathbf{\Lambda}_2\right\} \right) \\
&  \triangleq \text{Tr} \left(\mathbf{A}_m \mathbf{C}_m \right),
\end{aligned}
\end{equation}
The entry on the $k$-th row and $l$-th column of $\mathbf{C}_m$ equals
$$
\begin{aligned}
\mathbf{C}_m^{kl}&=\mathbb{E}\left\{e^{-j\alpha_{2k}}\mathbf{x}^H(k) \mathbf{B}_m \mathbf{x}(l)e^{j\alpha_{2l}}\right\} \\
&=\mathbb{E}\left\{e^{j(\alpha_{2l}-\alpha_{2k})}\text{Tr}\left(\mathbf{x}^H(k) \mathbf{B}_m \mathbf{x}(l)\right) \right\} \\
&=\mathbb{E}\left\{e^{j(\alpha_{2l}-\alpha_{2k})}\text{Tr}\left( \mathbf{B}_m \mathbf{x}(l)\mathbf{x}^H(k) \right) \right\} \\
&=e^{j(\alpha_{2l}-\alpha_{2k})}\text{Tr}\left( \mathbf{B}_m \mathbb{E}\left\{\mathbf{x}(l)\mathbf{x}^H(k) \right\} \right) \\
&=\delta_{kl}\text{Tr}\left( \mathbf{B}_m \mathbf{R}_{xl}\right),
\end{aligned}
$$
where $\delta_{kl}$ is the Kronecker delta function with value $1$ if $k=l$, and value $0$ otherwise. Therefore, matrix $\mathbf{C}_m$ is diagonal with $\text{Tr}\left(\mathbf{B}_m\mathbf{R}_{xl}\right)$ as its $l$-th entry. The interference power at the $m$-th radar receive antenna can be expressed as
\begin{equation} \label{eqn:intpowerIIm}
\begin{aligned}
\beta_m &= \text{Tr} \left(\mathbf{A}_m \text{diag}\left(\text{Tr}\left( \mathbf{B}_m \mathbf{R}_{x1}\right),\dots,\text{Tr}\left( \mathbf{B}_m \mathbf{R}_{xL}\right) \right)\right) \\
&=\sum\nolimits_{l=1}^L a_{l \xi_m} \text{Tr}\left( \mathbf{B}_m \mathbf{R}_{xl}\right)\\
&=\sum\nolimits_{l=1}^L a_{l \xi_m} \text{Tr}\left( \mathbf{g}^H_m \mathbf{R}_{xl} \mathbf{g}_m  \right) \\
&=\sum\nolimits_{l=1}^L a_{l \xi_m} \mathbf{g}^H_m \mathbf{R}_{xl} \mathbf{g}_m   \\
\end{aligned}
\end{equation}
where $a_{l \xi_m}$ denotes the $l$-th diagonal entry of $\mathbf{A}_m$ as defined in (\ref{eqn:EIPII1}).  Substituting $\beta_m$ in (\ref{eqn:intpowerIIm}) into (\ref{eqn:interfII}), we obtain the expression of the effective interference power to the Scheme II radar as follows
$$
\begin{aligned}
\text{EIP}_{II}&=\sum_{m=1}^{M_{r,R}}\beta_m = \sum_{m=1}^{M_{r,R}} \sum_{l=1}^L a_{l \xi_m } \mathbf{g}^H_m \mathbf{R}_{xl} \mathbf{g}_m  \\
&= \sum_{l=1}^L \sum_{m=1}^{M_{r,R}} a_{l \xi_m }  \mathbf{g}^H_m \mathbf{R}_{xl} \mathbf{g}_m  \\
&= \sum_{l=1}^L \text{Tr}\left( \mathbf{\Delta}_{l\xi}  \mathbf{G}_2 \mathbf{R}_{xl} \mathbf{G}^H_2 \right),
\end{aligned}
$$
which completes the proof.

\section{Proof of Lemma \ref{lemma:EIPII2}} \label{appen:EIPII2}
The $m$-th diagonal element of $\mathbf{\Delta}_{l\xi}$ can be expressed as
\begin{equation}
\begin{aligned}
a_{l\xi_m} &= \mathbf{s}_m^H(l)\mathbf{s}_m(l) = [\Delta_{\xi_m}\mathbf{s}(l)]^H\Delta_{\xi_m}\mathbf{s}(l)\\
& = \mathbf{s}^H(l)\Delta_{\xi_m}\mathbf{s}(l) = \text{Tr}\left(\mathbf{s}^H(l)\Delta_{\xi_m}\mathbf{s}(l)\right) \\
& = \text{Tr}\left(\Delta_{\xi_m}\mathbf{s}(l)\mathbf{s}^H(l)\right) = \mathbf{\Omega}_{m\cdot} \left(\mathbf{s}(l)\circ \mathbf{s}(l)\right),
\end{aligned}
\end{equation}
where $\Delta_{\xi_m} = \text{diag}(\mathbf{\Omega}_{m\cdot})$ and $\mathbf{\Omega}_{m\cdot}$ denotes the $m$-th row of $\mathbf{\Omega}_{II}$. Thus we have
$
\mathbf{\Delta}_{l\xi}= \text{diag}\left(\mathbf{\Omega}_{II} \left(\mathbf{s}(l)\circ \mathbf{s}(l)\right)\right)
$.
Substituting this into (\ref{eqn:EIPII1}) gives
\begin{equation}
\begin{aligned}
&\text{EIP}_{II} = \sum_{l=1}^L \text{Tr}\left( \text{diag}\left(\mathbf{\Omega}_{II} \left(\mathbf{s}(l)\circ \mathbf{s}(l)\right)\right) \mathbf{G}_{2}\mathbf{R}_{xl}\mathbf{G}^H_{2}\right) \\
& = \text{Tr}\left\{ \left[\mathbf{\Omega}_{II}(\mathbf{s}(1)\circ \mathbf{s}(1)),\dots,\mathbf{\Omega}_{II}(\mathbf{s}(L)\circ \mathbf{s}(L))\right]^T \mathbf{Q} \right\} \\
& = \text{Tr}\left\{ [\mathbf{\Omega}_{II}(\mathbf{S}\circ\mathbf{S})]^T \mathbf{Q}\right\} = \text{Tr}\left( \mathbf{\Omega}_{II}^T\mathbf{Q}(\mathbf{S}\circ\mathbf{S})^T\right).
\end{aligned}
\end{equation}
Lemma \ref{lemma:EIPII2} is proved.

\end{document}